%% file: arXiv_CR.tex
\documentclass[twoside]{article}

\usepackage{arxiv}

\usepackage{graphicx}
\usepackage{algorithm}
\usepackage{algpseudocode}
\usepackage{pdfpages}
\usepackage{hyperref}
\usepackage{url}
\usepackage{caption}
\usepackage{subcaption}
\usepackage{verbatim}
\usepackage[dvipsnames]{xcolor}
\usepackage{siunitx}
\usepackage{textalpha}
\usepackage{amsmath, amssymb, amsthm}

\theoremstyle{plain}
\newtheorem{proposition}{Proposition}[section]

\theoremstyle{definition}
\newtheorem{remarkinner}{Remark}[proposition]
\newenvironment{remark}
{
\begin{remarkinner}}
{\end{remarkinner}}

\input{math_commands.tex}

\newcommand{\valpm}[2]{$#1${\scalebox{0.7}{$\pm #2$}}}
\DeclareSIUnit[number-unit-product = {\,}]
\cal{cal}
\DeclareSIUnit
\bar{bar}

\begin{document}

\twocolumn[

\title{Split-Flows: Measure Transport and Information Loss Across Molecular Resolutions}

\author{ Sander Hummerich \And  Tristan Bereau \And Ullrich Köthe }
\runningauthor{Sander Hummerich, Tristan Bereau, Ullrich Köthe}

\address{ Institute for Theoretical Physics \\ Heidelberg University \And Institute for Theoretical Physics \\ Heidelberg University \And Computer Vision and Learning Lab \\ Heidelberg University }
]

\begin{abstract}
By reducing resolution, coarse-grained models greatly accelerate molecular simulations, unlocking access to long-timescale phenomena, though at the expense of microscopic information. Recovering this fine-grained detail is essential for tasks that depend on atomistic accuracy, making backmapping a central challenge in molecular modeling. We introduce split-flows, a novel flow-based approach that reinterprets backmapping as a continuous-time measure transport across resolutions. Unlike existing generative strategies, split-flows establish a direct probabilistic link between resolutions, enabling expressive conditional sampling of atomistic structures and—for the first time—a tractable route to computing mapping entropies, an information-theoretic measure of the irreducible detail lost in coarse-graining. We demonstrate these capabilities on diverse molecular systems, including chignolin, a lipid bilayer, and alanine dipeptide, highlighting split-flows as a principled framework for accurate backmapping and systematic evaluation of coarse-grained models. Our code is available at \url{https://github.com/BereauLab/split-flows}.

\end{abstract}

\section{INTRODUCTION}
\begin{figure}[ht]
    \centering
    \includegraphics[width=1\linewidth]{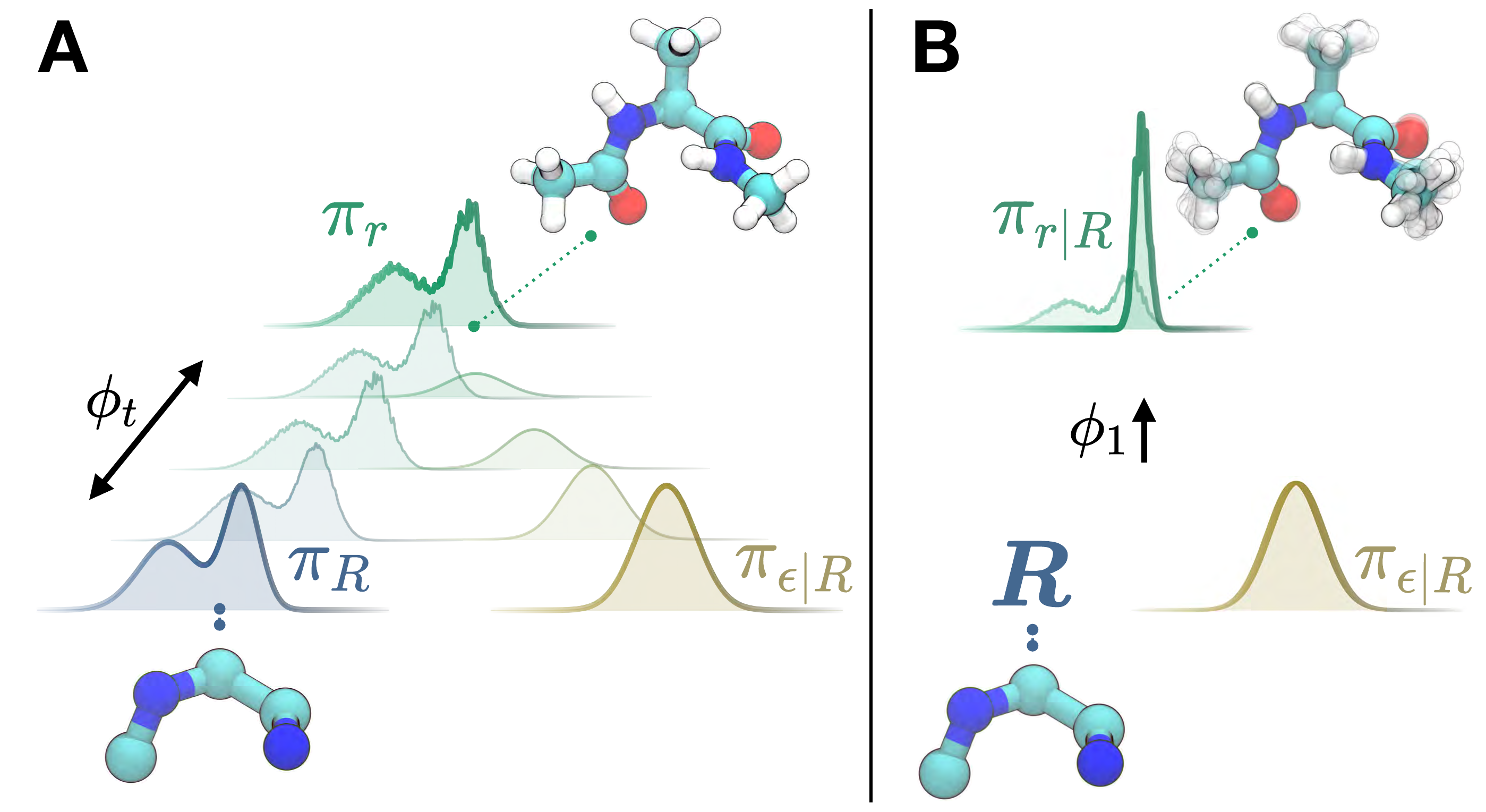}
    \caption{(A) Split-flows connect fine- and coarse-grained densities, $\pi_r$ and $\pi_R$, respectively, at different molecular resolutions via a continuous-time measure transport that maps the excess degrees of freedom of the fine-grained resolution to a simple noise distribution, $\pi_{\epsilon \mid R}$. (B) This enables sampling from the conditional density $\pi_{r \mid R}$, i.e., generative backmapping, and quantifies the information loss inherent in the coarse-grained representation.}
    \label{fig:figure-1}
\end{figure}

Coarse-grained models play a central role in molecular and material simulations \cite{Noid2023}. By marginalizing out unnecessary detail, they drastically reduce the computational cost of simulation and smooth out the underlying energy landscape. This enables simulations on length and time scales that are otherwise intractable in fine-grained models, providing an efficient tool to study slow collective dynamics and mesoscale phenomena such as protein folding, polymer conformational transitions, membrane remodeling, lipid-domain organization, and large-scale self-assembly.

A coarse-graining map implicitly defines an ill-posed inverse problem referred to as backmapping; that is, to reconstruct the marginalized degrees of freedom of the fine-grained model from the coarse-grained representation. As the forward process defines a many-to-one map—many detailed configurations are mapped to the same coarse-grained configuration—the reverse process can be cast as a generative-modeling problem: learning a probabilistic model for the distribution of fine-grained configurations corresponding to each coarse-grained representative.

The reduction of degrees of freedom in coarse-grained models inevitably leads to information loss relative to the fine-grained descreiption. This loss can be quantified through the concept of mapping entropy \cite{shell2008entropy, Foley2015}, which measures the average entropy of the distribution of fine-grained configurations that map to a given coarse-grained representative. Mapping entropy thus provides an information-theoretic lens on multiscale modeling, where a low mapping entropy corresponds to high information loss due to reduced resolution. This perspective enables quantitative assessment of the information loss in coarse-grained models and can ultimately inform model and simulation design.

In this work, we propose \textit{split-flows}---a novel flow-based model that provides a clear approach to bridging the dimensional gap between fine- and coarse-grained domains, as illustrated in Figure~\ref{fig:figure-1}. Split-flows define a continuous-time measure transport across dimensions, enabling us to connect the configurational densities at two different resolutions for general coarse-graining strategies. In addition to addressing the backmapping problem, this probabilistic link between fine- and coarse-grained resolutions allows us to compute the information loss of the coarse-graining map. In summary, we make the following contributions:
\begin{itemize}
    \item \textbf{\textit{Method:}} We introduce split-flows, a flow-based model that enables continuous-time transport of probability measures across different resolutions, bridging fine- and coarse-grained domains. 
    
    \item \textbf{\textit{Theory:}} We show that split-flows allow, for the first time, tractable and general computation of mapping entropy for arbitrary coarse-graining maps, providing a principled measure of information loss. 
    
    \item \textbf{\textit{Applications:}} We apply split-flows to diverse biomolecular systems—chignolin, a lipid bilayer, and alanine dipeptide—demonstrating accurate backmapping and their utility for information-theoretic assessment of coarse-grained models.
\end{itemize}

\section{RELATED WORK}
Solving the inverse problem of backmapping is a central challenge in multiscale molecular modeling \cite{Peter2009}. Mirroring trends across many scientific domains, data-driven methods increasingly replace traditional handcrafted algorithms, such as those by \cite{Rzepiela2010}, and \cite{Wassenaar2014}, which predict approximate fine-grained configurations from coarse inputs, followed by costly refinement.
Early approaches, such as \cite{Stieffenhofer2020}, \cite{Li2020}, and \cite{Wang2022}, leverage generative adversarial networks and variational autoencoders to generate fine-grained samples, without the need for post hoc refinement. \cite{Shmilovich2022} extend this line of work by incorporating information along reconstructed simulations to ensure temporal consistency.
More recent methods by \cite{Jones2023}, \cite{ferguson2025flowback, Berlaga2025}, and \cite{UgarteLaTorre2025} adopt multi-step samplers, i.e., continuous normalizing flows and diffusion models, enabling generalization to unseen structures through residue-wise processing and transferable coarse-graining schemes. While these models emphasize energetic plausibility, transferability, or dynamical consistency, they do not establish a probabilistic link between resolutions and therefore miss key statistical properties of the coarse-graining map. Our method addresses this limitation.

Normalizing flows, introduced by \cite{rezende2015flows} in discrete form, map complex data distributions to simple latents. The continuous-time formulation of \cite{Chen2018} improves expressiveness but initially lacks a tractable training procedure. Flow matching \cite{Lipman2023} resolves this by replacing maximum likelihood with a quadratic regression objective for the underlying velocity field, enabling efficient training of continuous normalizing flows. Modern formulations of flow matching, particularly those by \cite{Albergo2023} and \cite{Tong2024}, generalize normalizing flows to define a measure transport between arbitrary pairs of distributions. Most similar to our approach, \cite{Albergo2023data} apply this framework to image super-resolution and in-painting. We build on this modern interpretation of continuous normalizing flows to connect molecular configurations across resolutions. 

Mapping entropy, introduced by \cite{shell2008entropy} and subsequently formalized by \cite{Rudzinski2011,Foley2015}, quantifies the information lost when a fine-grained system is represented at reduced resolution. As such, it provides a principled criterion for analyzing coarse-grained representations. Prior work has used mapping entropy to characterize the entropic structure of coarse-grained models \cite{Kidder2021}, to identify informative mappings for specific systems such as actin \cite{Kidder2024}, and to study the extent to which structural reduction preserves dynamical behavior across resolutions \cite{Armstrong2012,Jin2023}. It has also been used to disentangle entropic and energetic contributions to collective variables \cite{Mussi2025}, implemented in practical coarse-graining software frameworks \cite{Giulini2020,Giulini2024}, and applied beyond molecular modeling to settings such as spin systems and low-dimensional financial descriptors \cite{Holtzman2022}. Despite this breadth of applications, existing approaches are often tailored to particular models or system classes. By contrast, split-flows provide a general and rigorous framework for computing mapping entropies across a broad class of systems and reduction strategies.

\section{PRELIMINARIES}
\subsection{Thermodynamic Framework}
\textbf{\textit{Notation:}} We use lowercase variable names to denote quantities at the fine-grained resolution and uppercase variable names for quantities at the coarse-grained resolution. For ease of presentation, we treat variables as unitless quantities. 

We consider a system with $n$ degrees of freedom at temperature $T$ and configurations denoted by $\vr$. At equilibrium, these configurations follow the Boltzmann distribution governed by the potential energy function $u : \mathbb{R}^n \to \mathbb{R}$:  
\begin{equation} \label{eq:boltzmann-distribution}
    \pi_r(\vr) = Z^{-1} \exp \left[- u(\vr)/(k_\mathrm{B} T)\right],
\end{equation}
where $Z = \int_{\mathbb{R}^n} \mathrm{d}\vr \exp[-u(\vr)/(k_\mathrm{B} T)]$ is the normalization constant and $k_\mathrm{B}$ is the Boltzmann constant.  

In practice, samples from $\pi_r$ are generated using trajectory-based methods such as molecular dynamics or Monte Carlo simulations. These methods often suffer from slow convergence at the fine-grained resolution, since the energy landscape is rugged and trajectories can become trapped in local minima. Coarse-grained models accelerate sampling by both reducing the number of degrees of freedom and smoothing out the underlying energy surface \cite{Noid2013}.

\subsection{Coarse-Graining} \label{sec:coarse-graining}
\begin{figure}[ht]
    \centering
    \includegraphics[page=1,width=1\linewidth]{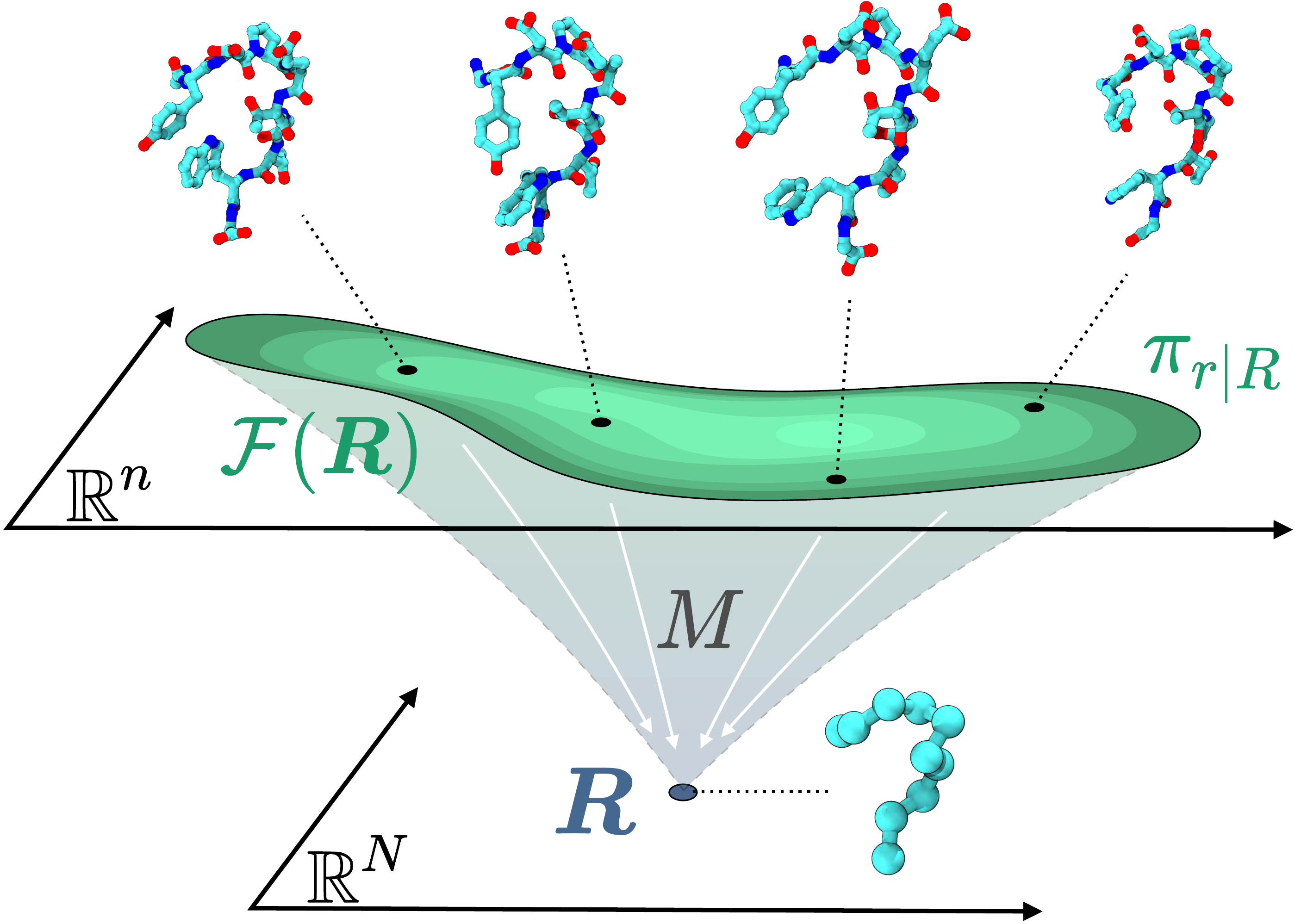}
    \caption{Bottom-up coarse-graining defines a many-to-one mapping operator $M$ that reduces a set $\mathcal{F}(\mR)$ of fine-grained configurations to a single coarse-grained representative $\mR$.}
    \label{fig:fiber-illustration}
\end{figure}
In this work, we focus on so-called bottom-up coarse-graining approaches that derive a coarse representation from a fine-grained model via a coarse-graining map
\begin{equation} \label{eq:linear-cg-map}
    M: \mathbb{R}^n \rightarrow \mathbb{R}^N, \qquad \mR = M(\vr),
\end{equation}
which assigns to each fine-grained configuration $\vr$ a coarse-grained representative $\mR$ with $N$ degrees of freedom, as shown in Figure~\ref{fig:fiber-illustration}. Such maps aim to preserve the essential physics, effectively separating \textit{slow} (typically complex) from \textit{fast} (typically simple) degrees of freedom.

The corresponding coarse-grained density $\pi_R$ is obtained by integrating out the fast degrees of freedom:
\begin{equation} \label{eq:explicit-marginal}
\pi_R(\mR) = \int_{\mathbb{R}^n} \mathrm{d}\vr \, \delta(M(\vr) - \mR) \; \pi_r(\vr),
\end{equation}
where the delta function ensures that only fine-grained configurations consistent with $\mR$ contribute. Since this exact density is generally intractable, coarse-graining methods approximate $\pi_R$ with a model $\hat{\pi}_R$ that ideally preserves consistency with the above equation.

In this work, however, we assume access to the exact coarse-grained distribution. Analogous to the fine-grained Boltzmann distribution, it can be written as
\begin{equation} \label{eq:cg-density-pmf}
\pi_R(\mR) \propto \exp \left[-W(\mR)/(k_\mathrm{B} T)\right].
\end{equation}
Here, the effective potential $W: \mathbb{R}^N \to \mathbb{R}$ is a \textit{free energy},
\begin{equation} \label{eq:pmf-decomposition}
W(\mR) = E(\mR) - T S(\mR),
\end{equation}
which includes energetic and entropic contributions \cite{Noid2013}; $E(\mR)$ is the mean fine-grained energy conditioned on $\mR$, and $S(\mR)$ is an entropic term, that measures the structure of the distribution of compatible fine-grained configurations. Coarse-graining averages over microscopic energies while introducing an entropic bias toward states with many realizations. This results in a smoother free-energy landscape $W(\mR)$ that is easier to sample than the atomistic potential, at the cost of information loss, as different fine-grained configurations mapping to the same $\mR$ become indistinguishable.

\subsection{Information Loss in Coarse-Grained Representations} \label{sec:information-loss}

A quantitative measure of information loss in coarse-grained representations can be derived from the concept of mapping entropy $S_\mathrm{map}$ and its configuration-dependent (\textit{local}) counterpart $S(\mR)$. 

To introduce the mapping entropy, we first define the \textit{fiber} associated with a coarse-grained representative. The fiber is the pre-image of $\mR$ under the mapping $M$, i.e., it is the \textit{lost subensemble} \cite{Kidder2024} of all fine-grained states that map to $\mR$:
\begin{equation} \label{eq:fiber-definition}
    \mathcal{F}(\mR) = \{\vr \in \mathbb{R}^n \;|\; M(\vr) = \mR\}.
\end{equation}
Bayes' theorem gives the \textit{fiber distribution}---the conditional probability of a fine-grained configuration given its coarse-grained representative:
\begin{equation}\label{eq:fiber-bayes}
    \pi_{r \mid R}(\vr \mid \mR) = \frac{\pi_r(\vr)}{\pi_R(\mR)}, \, \forall \, \vr \in \mathcal{F}(\mR).
\end{equation}
We will denote the expectation of some $d$-dimensional observable $O: \mathbb{R}^n \rightarrow \mathbb{R}^d$ on the fine-grained configuration space as the \textit{fiber average}:
\begin{equation} \label{eq:fiber-average}
    \begin{aligned}
        \mathbb{E}_{r \mid R} [O(\vr)] 
        &= \int_{\mathcal{F}(\mR)} \mathrm{d}\vr \; \pi_{r \mid R}(\vr \mid \mR) \, O(\vr) \\
    \end{aligned}
\end{equation}
which lets us evaluate observables over fine-grained states consistent with one particular coarse-grained representative, e.g., the energetic component $E(\mR) = \mathbb{E}_{r \mid R} [u(\vr)]$ in Equation~\ref{eq:pmf-decomposition}.

Using Equation~\ref{eq:fiber-bayes}, we can write the entropy of the fiber distribution as
\begin{equation} \label{eq:conditional-mapping-entropy}
    \begin{aligned}
        S(\mR) &= - k_\mathrm{B} \int_{\mathcal{F}(\mR)} \mathrm{d}\vr \; \pi_{r \mid R}(\vr \mid \mR) \log \pi_{r \mid R}(\vr \mid \mR) \\
        &= - k_\mathrm{B} \,  \mathbb{E}_{r \mid R} \left [ \log \frac{\pi_r(\vr)}{\pi_R(\mR)} \right ],
    \end{aligned}
\end{equation}
which we denote the \textit{local mapping entropy}. As outlined in Appendix~\ref{app:pmf-decomposition}, this is the entropic contribution $S(\mR)$ in Equation~\ref{eq:pmf-decomposition}. For compact domains, e.g., a periodic box, we can define the local excess mapping entropy as the relative entropy of the fiber distribution compared to the best guess we can make without any prior information, i.e., a uniform distribution over $\mathcal{F}(\mR)$:
\begin{equation} \label{eq:excess-conditional-mapping-entropy}
    \begin{aligned}
        S_\mathrm{e}(\mR) 
        & = - k_\mathrm{B} \, \mathbb{E}_{r \mid R} \left [ \log \frac{\pi_{r\mid R}(\vr\mid \mR)}{u_{r \mid R}(\vr \mid \boldsymbol{R})} \right ] \\
    \end{aligned}
\end{equation}
which is the Kullback-Leibler divergence between the fiber distribution $\pi_{r \mid R}$ and a uniform distribution with density \(u_{r \mid R}(\vr \mid \boldsymbol{R}) = \mathrm{Vol}(\mathcal{F}(\mR))^{-1} \; \forall \, \vr \in \mathcal{F}(\mR)\) defined over the fiber.

The local (excess) information loss due to reducing the fiber $\mathcal{F}(\mR)$ to a single representative $\mR$ in the coarse-grained model relates to the local (excess) mapping entropy as
\begin{equation} \label{eq:information-loss}
    I_{(e)}(\mR) = - S_{(e)}(\mR) / k_\mathrm{B}.
\end{equation}
It is evident that the local information loss $I(\mR)$ must be non-negative and thus $S(\mR) \leq 0$. Taking the expectation of the local quantities $S(\mR)$ and $I(\mR)$ with respect to the coarse-grained density $\pi_R$ then yields the global mapping entropy $S_\mathrm{map}$ and information loss $I_\mathrm{map}$ of the coarse-grained model.

\subsection{Two-Sided Flow Matching} \label{sec:two-sided-fm}
Two-sided flow matching aims to connect two non-trivial distributions $\pi_0$ and $\pi_1$ over an interpolation interval $[0, 1]$. Continuous normalizing flows (CNFs) define such a measure transport via the solution to an ordinary differential equation (ODE):
\begin{equation} \label{eq:cnf}
    \begin{aligned}
        \frac{\mathrm{d}}{\mathrm{d} t} \phi_t(\vx_0) =  \vv_t^\theta(\phi_t(\vx_0)), && \phi_0(\vx_0) = \vx_0.
    \end{aligned}
\end{equation}
Here, $\vv^\theta: [0, 1] \times \mathbb{R}^n \rightarrow \mathbb{R}^n$ is a time-dependent velocity field, which is parameterized by a neural network. The flow defines a continuous-time bijection between samples from the two endpoint distributions, $\pi_0$ and $\pi_1$. The pushforward of the initial density $\pi_0$ under the flow $\phi_t$ is given by
\begin{equation} \label{eq:change-of-variables}
   \log \pi_t(\phi_t(\vx_0)) = \log \pi_0(\vx_0) - \int_0^t \mathrm{d}\tau \; \nabla \cdot \vv_\tau ^\theta(\phi_\tau(\vx_0)),
\end{equation}
which defines a probability path between $\pi_0$ and $\pi_1$.

Given a coupling $\pi_{0, 1}$ of samples of two endpoint distributions $\pi_0$ and $\pi_1$, \cite{Albergo2023} propose the following quadratic regression objective:
\begin{equation} \label{eq:conditional-flow-matching-loss}
    \mathcal{L}_v(\theta) = \int _0^1 \mathrm{d}t \; \mathbb{E}_{0, 1}[\|\vv_t^\theta(I_t(\vx_0, \vx_1)) - \partial_t I_t(\vx_0, \vx_1)\|^2],
\end{equation}
which is a simple extension of the conditional flow matching objective, originally introduced by \cite{Lipman2023}. The coupling $\pi_{0, 1}$ defines how the flow should pair samples from the two endpoint distributions and is task-specific, e.g., an optimal transport coupling. It satisfies $\int_{\mathbb{R}^n} \mathrm{d}\vx_1 \; \pi_{0, 1}(\vx_0, \vx_1) = \pi_0(\vx_0)$ and $\int_{\mathbb{R}^n} \mathrm{d}\vx_0 \; \pi_{0, 1}(\vx_0, \vx_1) = \pi_1(\vx_1)$.
The interpolant $I_t$ is chosen to be of the form $I_t(\vx_0, \vx_1) = \alpha_t \vx_0 + \beta_t \vx_1$ and obeys the boundary conditions $\alpha_0 = \beta_1 = 1$ and $\alpha_1 = \beta_0 = 0$.

\section{SPLIT-FLOWS} \label{sec:split-flows}
\begin{figure*}[ht] 
    \centering
    \includegraphics[page=1,width=0.8\linewidth]{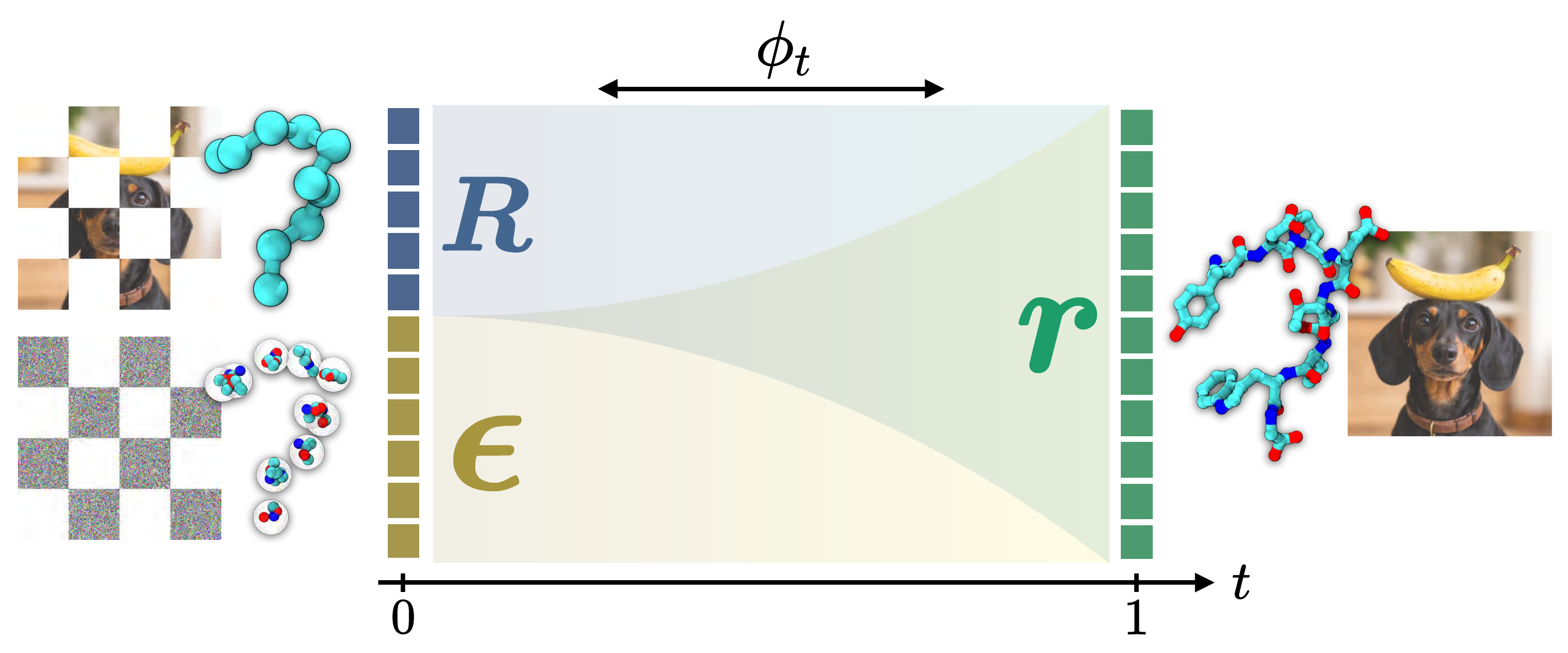}
    \caption{Split-flows define a one-to-one map between configurations of different resolutions. The lower-dimensional samples $\mR$ are augmented with noise $\boldsymbol{\epsilon}$ to resolve the degeneracy induced by the dimensionality gap. The flow $\phi_t$ connects the augmented coarse-grained configurations $(\mR, \boldsymbol{\epsilon}) \sim \pi_R \times \pi_{\epsilon \mid R}$ at $t=0$ with the fine-grained configurations $\vr \sim \pi_r$ at $t=1$. An instructive analogy arises in image inpainting: a partially observed image is augmented with noise dimensions, and the split-flow acts as a probabilistic bridge to a complete, coherent one.} \label{fig:split-flow-illustration}
\end{figure*}

\textbf{\textit{Notation:}} We identify the endpoint densities as $\pi_0 = \pi_R \times \pi_{\epsilon \mid R}$, an augmented coarse-grained density, and $\pi_1 = \pi_r$, the fine-grained density. Correspondingly, $\vx_0 = (\mR, \boldsymbol{\epsilon})$ and $\vx_1 = \vr$. The augmented coarse-grained density $\pi_R \times \pi_{\epsilon \mid R}$ is introduced below.

Split-flows bridge the gap between distributions defined over domains with different dimensionality by augmenting the lower-dimensional space with additional noise dimensions, as illustrated in Figure~\ref{fig:split-flow-illustration}. Given the two endpoint distributions $\pi_R$ and $\pi_r$ defined over $\mathbb{R}^N$ and $\mathbb{R}^n$, respectively, we introduce a simple noise distribution $\pi_{\epsilon \mid R}$ on $\mathbb{R}^{n-N}$ and use a CNF $\phi_t$, trained via the conditional flow matching objective in Equation~\ref{eq:conditional-flow-matching-loss}, to learn a measure transport between $\pi_R \times \pi_{\epsilon \mid R}$ and $\pi_r$:
\begin{equation}
\begin{aligned}
\phi_1: \mathbb{R}^N \times \mathbb{R}^{n-N} \rightarrow \mathbb{R}^n, && (\mR, \boldsymbol{\epsilon}) \mapsto \phi_1(\mR, \boldsymbol{\epsilon}) = \vr.
\end{aligned}
\end{equation}
The noise distribution $\pi_{\epsilon \mid R}$ is chosen such that, given a coarse-grained representation, sampling is tractable, e.g., a Gaussian distribution. Using Equation~\ref{eq:change-of-variables} and the factorization of the augmented endpoint distribution, we can connect the densities $\pi_R$ and $\pi_r$ despite the difference in dimensionality via
\begin{equation} \label{eq:change-of-variables-split}
\begin{aligned}
\log \pi_r(\phi_1(\mR, \boldsymbol{\epsilon})) = & \log \pi_R(\mR) + \log \pi_{\epsilon \mid R}(\boldsymbol{\epsilon} \mid \mR) \\
& - \int_0^1 \mathrm{d}\tau \; \nabla \cdot \vv_\tau^\theta(\phi_\tau(\mR, \boldsymbol{\epsilon})).
\end{aligned}
\end{equation}

In molecular modeling, we use this framework to relate the coarse-grained density $\pi_R$ to the density over fine-grained configurations $\pi_r$. Introducing the conditional noise distribution $\pi_{\epsilon \mid R}$ resolves the many-to-one nature of coarse-graining and allows backmapping to be formulated as transport of measures across resolutions. From a geometric perspective, the flow learns a global coordinate transformation that disentangles the structure of fine-grained configuration space induced by the map $M$, i.e., its decomposition into slow and fast degrees of freedom. A more detailed discussion of this viewpoint is provided in Appendix~\ref{app:geometric-perspective}.

\begin{algorithm} 
    \caption{Per-sample loss computation} \label{alg:split-flow-loss}
    \begin{algorithmic}[1]
        \State \textbf{Input:} fine-grained configuration $\vr$, velocity field $\vv^\theta$,  coarse-graining map $M$, noise distribution $\pi_{\epsilon \mid R}$, interpolant $I_t$
        \State Compute CG representation: $\mR \gets M(\vr)$
        \State Sample noise: $\boldsymbol{\epsilon} \sim \pi_{\epsilon \mid R}$
        \State Sample time: $t \sim u_{[0, 1]}$
        \State Compute loss: 
        \[
            \mathcal{L}(\theta, \vr) \gets\|\vv_t^\theta(I_t(\mR, \boldsymbol{\epsilon}, \vr)) - \partial_t I_t(\mR, \boldsymbol{\epsilon}, \vr)\|^2
        \]
        \State \textbf{Output}: Per-sample loss $\mathcal{L}(\theta, \vr)$
    \end{algorithmic}
\end{algorithm}
To train split-flows in a two-sided manner, as outlined in Section~\ref{sec:two-sided-fm}, we pair samples from the two endpoint distributions using the coarse-graining map $M$, and construct a semi-deterministic coupling between $(\mR, \boldsymbol{\epsilon})$ and $\vr$:
\begin{equation} \label{eq:deterministic-coupling}
\pi_{R, \epsilon, r}(\mR, \boldsymbol{\epsilon}, \vr) = \pi_r(\vr) \, \delta(\mR - M(\vr)) \, \pi_{\epsilon \mid R}(\boldsymbol{\epsilon}\mid \mR).
\end{equation}
This coupling encourages the flow to correctly pair fine-grained configurations with their respective coarse-grained counterparts and provides a straightforward way to evaluate a Monte Carlo estimate of the objective in Equation~\ref{eq:conditional-flow-matching-loss}. We outline the per-sample loss computation in Algorithm~\ref{alg:split-flow-loss}.

This setup, once trained, allows us to easily access the fibers, i.e., the many possible fine-grained configurations mapping to a single coarse-grained representative, and the local mapping entropy of the coarse-graining map. We can generate samples $\vr \mid \mR$ from the conditional distribution $\pi_{r \mid R}$, i.e., samples on the fiber $\mathcal{F}(\mR)$, using Algorithm~\ref{alg:split-flow-fiber-sampling}.
\begin{algorithm} 
    \caption{Fiber-constrained sampling} \label{alg:split-flow-fiber-sampling}
    \begin{algorithmic}[1]
        \State \textbf{Input:} coarse-grained configuration $\mR$, velocity field $\vv^\theta$, noise distribution $\pi_{\epsilon \mid R}$
        \State Sample noise: $\boldsymbol{\epsilon} \sim \pi_{\epsilon \mid R}$
        \State Define: $\vx_0 = \begin{bmatrix} \mR & \boldsymbol{\epsilon} \end{bmatrix}^\top$
        \State Numerically solve Equation~\ref{eq:cnf}: \[
        \vx_1 = \vx_0 + \int_0^1 \mathrm{d}\tau \; \vv^\theta_\tau(\phi_\tau(\vx_0))
        \]
        \State \textbf{Output}: Sample on fiber $\vr = \vx_1 \in \mathcal{F}(\mR)$
    \end{algorithmic}
\end{algorithm}

As outlined in Appendix~\ref{app:fiber-averages}, we can write the fiber average of an observable $O: \mathbb{R}^n \rightarrow \mathbb{R}^d$ using Equation~\ref{eq:fiber-average}, as
\begin{equation} \label{eq:fiber-average-split-flow}
\mathbb{E}_{r \mid R} [O(\vr)] = \mathbb{E}_{\epsilon \mid R} [O(\phi_1(\mR, \boldsymbol{\epsilon}))].
\end{equation}
By combining Equations~\ref{eq:change-of-variables-split} and \ref{eq:fiber-average-split-flow}, we can use split-flows to obtain an estimate of the local mapping entropy in Equation~\ref{eq:conditional-mapping-entropy}:
\begin{equation} \label{eq:local-mapping-entropy-split-flow}
\begin{aligned}
S (\mR) = & - k_\mathrm{B} \, \mathbb{E}_{\epsilon \mid R} \left[ \log \pi_{\epsilon \mid R}(\boldsymbol{\epsilon} \mid \mR) \right] \\
& + k_\mathrm{B} \, \mathbb{E}_{\epsilon \mid R} \left[ \int_0^1 \mathrm{d}\tau \; \nabla \cdot \vv_\tau^\theta(\phi_\tau(\mR, \boldsymbol{\epsilon})) \right],
\end{aligned}
\end{equation}
which requires evaluating the entropy of the noise distribution and the volume change under the flow.

\section{EXPERIMENTS}

In the experimental section, we present split-flows across three molecular systems. First, we consider the mini-protein chignolin, where we validate its backmapping capabilities and quantify the information loss along a coarse-grained molecular dynamics (MD) trajectory. Next, we investigate information loss in a coarse-grained representation of a two-particle solute dragged through a lipid membrane. Finally, we present the information loss landscape in the Ramachandran representation of alanine dipeptide.

\subsection{Chignolin} \label{sec:chignolin}
We apply split-flows to chignolin, a protein composed of 10 amino acids and 77 heavy atoms, which, despite its manageable size, already exhibits folding behavior. The coarse-graining map reduces the fine-grained configuration to the 10 $C_\alpha$ atoms, a commonly used reduction for proteins, as depicted in Figure~\ref{fig:split-flow-illustration}.

We train split-flows on 50k frames from a $1\unit{\micro\second}$ atomistic MD simulation at $360\ \unit{\kelvin}$, which includes multiple folding and unfolding transitions. Since split-flows operate on Cartesian coordinates, we use the $E(3)$-equivariant graph neural network (GNN) architecture proposed by \cite{Satorras2021}. As the noise distribution $\pi_{\epsilon \mid R}$, we choose a residue-wise Gaussian distribution centered at the position of the respective $C_\alpha$ atom. A detailed description and the hyperparameters for both the simulation and the model are provided in Appendix~\ref{app:chignolin}.

\begin{figure*}[ht]
    \centering
    \includegraphics[width=1\linewidth]{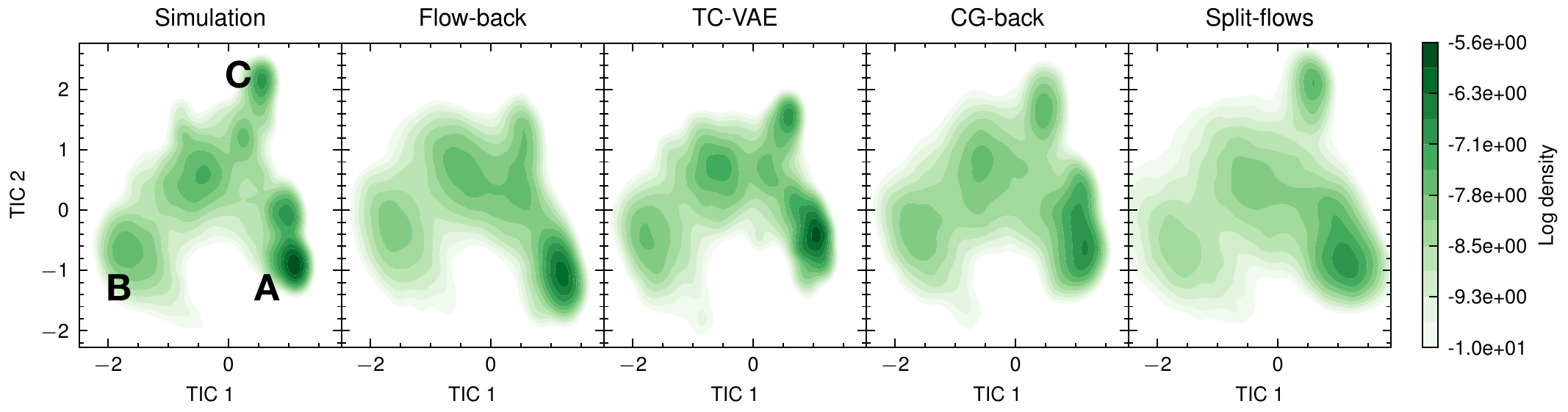}
    \caption{Log densities in the plane of the first two components of TICA. We present the projected log densities of the original simulated configurations as well as backmapped configurations using reference methods and our split-flows. The projection separates the folded (A), unfolded (B), and misfolded (C) modes of chignolin.}
    \label{fig:chignolin-tica-fes}
\end{figure*}
\textbf{\textit{Backmapping:}} First, we validate split-flows by means of backmapping. Given a test set of atomistic and coarse-grained configurations, $\vr$ and $\mR = M(\vr)$, we sample reconstructions $\hat{\vr} = \phi_1(\mR, \boldsymbol{\epsilon})$ with $\boldsymbol{\epsilon} \sim \pi_{\epsilon \mid R}$. We compare split-flows to three methods: TC-VAE \cite{Shmilovich2022}, Flow-back \cite{ferguson2025flowback}, and CG-back \cite{UgarteLaTorre2025}. Flow-back and CG-back transfer to unseen molecules via residue-based backmapping of the $C_\alpha$-representation, so we use the authors’ pretrained models. We retrain TC-VAE with the released code and hyperparameters. 

To evaluate structural fidelity, we project the fine-grained configurations onto the first two components of a time-lagged independent component analysis (TICA) \cite{Prez-Hernndez2013}, a commonly used projection. In Figure~\ref{fig:chignolin-tica-fes}, we compare the resulting log-densities in this two-dimensional representation. We find that split-flows, aside from slight smoothing, reproduce all major modes of the original density—especially the misfolded state, which is typically underrepresented by other methods—indicating high diversity of backmapped samples.

In Table~\ref{tab:chignolin-benchmark}, we report several numerical metrics. We measure energetic plausibility by computing the Wasserstein-1 distance $W_1$ between the distributions of internal energies of the original and reconstructed configurations. We assess the consistency of backmapped configurations with their initial coarse-grained counterparts by calculating the root-mean-squared deviation (RMSD) between the original coarse-grained representation $\mR = M(\vr)$ and the projected backmapped configuration $\hat{\mR} = M(\hat{\vr})$, denoted by $\mathrm{RMSD}_\mathrm{cg}$. To evaluate topological agreement, we construct a molecular graph based on atomic distances and compute the relative graph edit distance $D_\mathcal{G}$ with respect to the true molecular graph. For these three metrics we report mean and standard deviation over five test trajectories, each containing 10k frames.

To measure diversity within a fiber, we generate a set of configurations $\hat{\vr}_i \in \mathcal{F}(M(\vr))$ for a given reference structure $\vr$, and compute the average RMSD between generated configurations and the reference, denoted as $\mathrm{RMSD}_\mathrm{ref}$, as well as the average pairwise RMSD between the generated configurations, denoted as $\mathrm{RMSD}_\mathrm{gen}$. Following \cite{Jones2023}, we define a diversity score $\eta_\mathrm{div}$ as the ratio $\mathrm{RMSD}_\mathrm{gen} / \mathrm{RMSD}_\mathrm{ref}$. This ratio vanishes for deterministic backmapping, where all generated samples are identical, and increases with sample diversity. We report mean and standard deviation over 50 reference configurations, with 1k samples on the fiber per reference.
\begin{table}[ht!]
\caption{We report the Wasserstein-1 distance $\mathrm{W}_1$ of the internal energy distribution in $\unit{\kilo\cal\per\mol}$, the RMSD in the coarse-grained space $\mathrm{RMSD}_\mathrm{cg}$ in $\unit{\pico\meter}$, the relative graph edit distance $D_\mathcal{G}$ in $\%$, and the fiber-diversity score $\eta_\mathrm{div}$. Best values are highlighted in bold, second-best values are underlined.}  
\label{tab:chignolin-benchmark}
\centering
\resizebox{\linewidth}{!}{
\begin{tabular}{lcccc}
\textbf{Model} & $\boldsymbol{W_1 (\downarrow)}$ & $\boldsymbol{\mathrm{RMSD}_\mathrm{cg} (\downarrow)}$ & $\boldsymbol{D_\mathcal{G} (\downarrow)}$ & $\boldsymbol{\eta_\mathrm{div} (\uparrow)}$\\
\hline \\
Flow-back & \underline{\valpm{300}{3}} & \valpm{4.602}{0.008} & \valpm{\mathbf{0.027}}{\mathbf{0.006}} & \valpm{0.60}{0.10}\\
TC-VAE & \valpm{5900}{150} & \valpm{5.0}{0.3} & \valpm{6.2}{0.9} & \valpm{0.022}{0.005} \\
CG-back & \valpm{321}{3} & \valpm{\mathbf{0.071}}{\mathbf{0.007}}  & \valpm{0.53}{0.07} & \valpm{\mathbf{0.90}}{\mathbf{0.09}} \\
Split-flows & \valpm{\mathbf{131.1}}{\mathbf{1.8}} & \underline{\valpm{0.62}{0.04}} & \underline{\valpm{0.22}{0.06}} & \underline{\valpm{0.79}{0.15}}\\
\end{tabular}
}
\end{table}

Across all numerical metrics in Table~\ref{tab:chignolin-benchmark}, we find that split-flows perform competitively compared to existing methods. In particular, their ability to compute highly diverse samples—with a diversity score of $0.79$—that are simultaneously energetically plausible, with a Wasserstein-1 distance of $131.1 \unit{\kilo \cal \per \mol}$, places split-flows in a prominent position in the comparison. Moreover, split-flows rank second in terms of coarse-grained consistency and relative graph edit distance. Nonetheless, we emphasize that Flow-back and CG-back exhibit consistently strong performance, despite their transferability. We note that the low diversity score for TC-VAE results from the model’s coherency with the previous atomistic configuration, which limits generated configurations to remain temporally consistent with their respective predecessors. Furthermore, despite granting it a multiple of the training time used for our method, we find that TC-VAE does not perform as well as presented in the original work. More details on training and the compute budget used can be found in Appendix~\ref{app:chignolin}.

\textbf{\textit{Information loss:}}
Next, we leverage the mapping entropy framework developed in Sections~\ref{sec:information-loss} and \ref{sec:split-flows} to quantify the local information loss of the coarse-grained representation along a MD trajectory. We compute the information loss over a short section of a test trajectory and visualize the resulting sequence in Figure~\ref{fig:information-loss-chignolin}.

\begin{figure}[ht]
    \centering
    \includegraphics[page=1,width=1\linewidth]{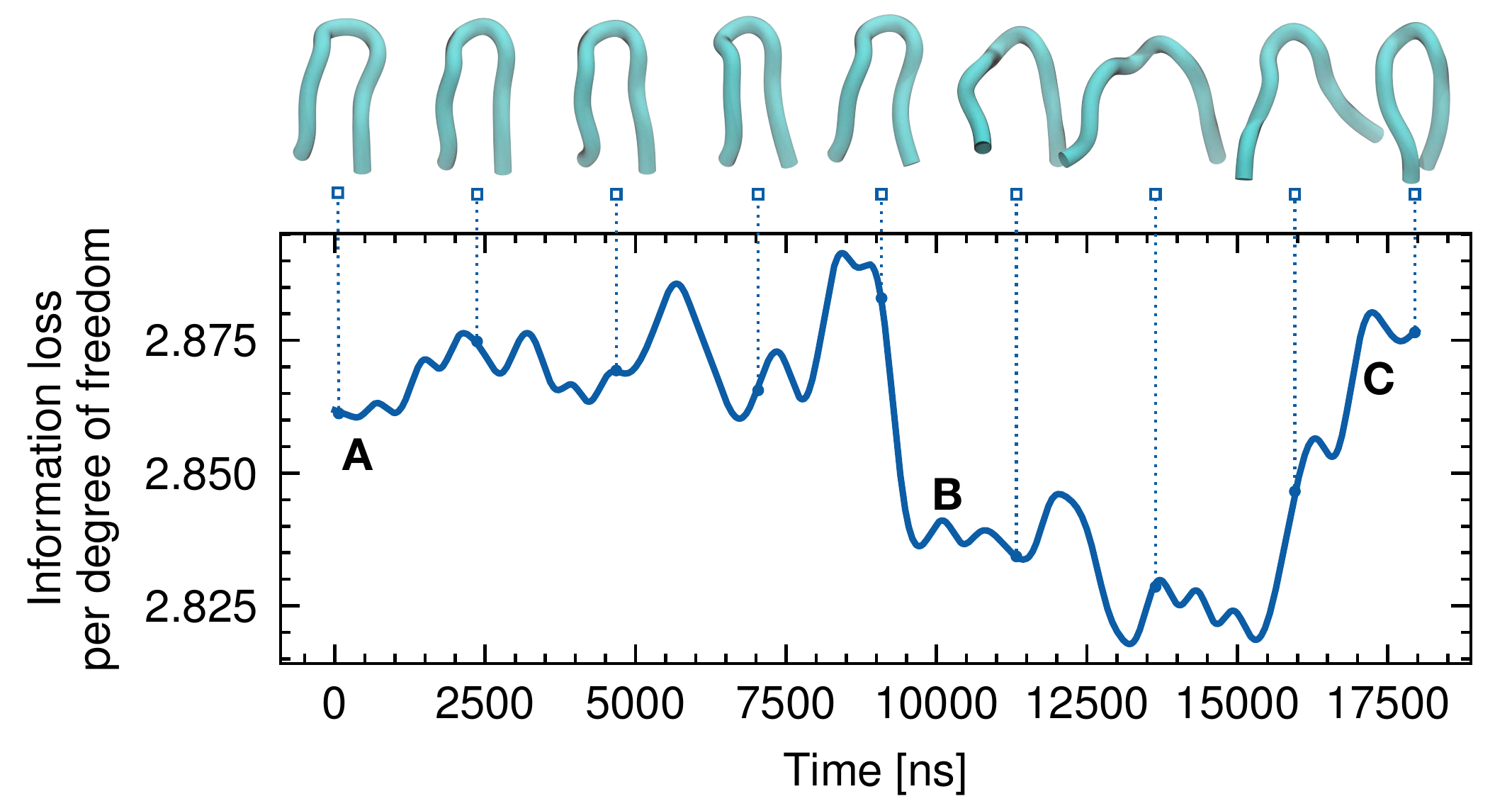}
    \caption{Average information loss per removed degree of freedom in the $C_\alpha$ representation of chignolin along a MD trajectory. We analyze a short section of the simulation starting in a folded state (A), followed by a partial separation of the two strands (B), and returning to the folded state (C). }
    \label{fig:information-loss-chignolin}
\end{figure}

The reduction to the $C_\alpha$ atoms, as depicted in Figure~\ref{fig:split-flow-illustration}, projects out many orthogonal degrees of freedom, particularly in the side chains. The associated removal of interactions leads to a conformation-dependent information loss: We observe a drop in the information loss landscape in the region where the two strands of the protein separate. This partial opening reduces the interactions between the projected-out atoms in the tails, resulting in a less constrained—and therefore less informative—fiber distribution.

\subsection{Solute in a Lipid Bilayer}\label{sec:lipid-bilayer}

Since the solute is approximately a rigid body, its configuration is well-described by a two-dimensional description consisting of the distance $z \in \left[-\tfrac{L}{2}, \tfrac{L}{2}\right]$  of the solute’s center of mass from the membrane center and the relative orientation $\vartheta \in [0, \pi]$ with respect to the $z$-axis, as depicted in Figure~\ref{fig:lipid-bilayer}. We define a coarse-grained description of the solute by projecting out the rotational degree of freedom: $M: \left[-\tfrac{L}{2}, \tfrac{L}{2}\right] \times [0, \pi] \rightarrow \left[-\tfrac{L}{2}, \tfrac{L}{2}\right]$.
We then train a split-flow, parameterized by a multilayer perceptron (MLP), to connect the configurations $\vr = \begin{bmatrix} z & \vartheta \end{bmatrix}^\top \quad \text{and} \quad \mR = \begin{bmatrix} z \end{bmatrix}^\top$ and use a uniform distribution on \( [0, \pi] \), \( \pi_{\epsilon \mid R} = u_{[0, \pi]} \), for the noise dimensions. Periodicity is enforced by a simple sine–cosine input parameterization for the MLP.

\begin{figure}[ht]
    \centering
    \includegraphics[page=1,width=1\linewidth]{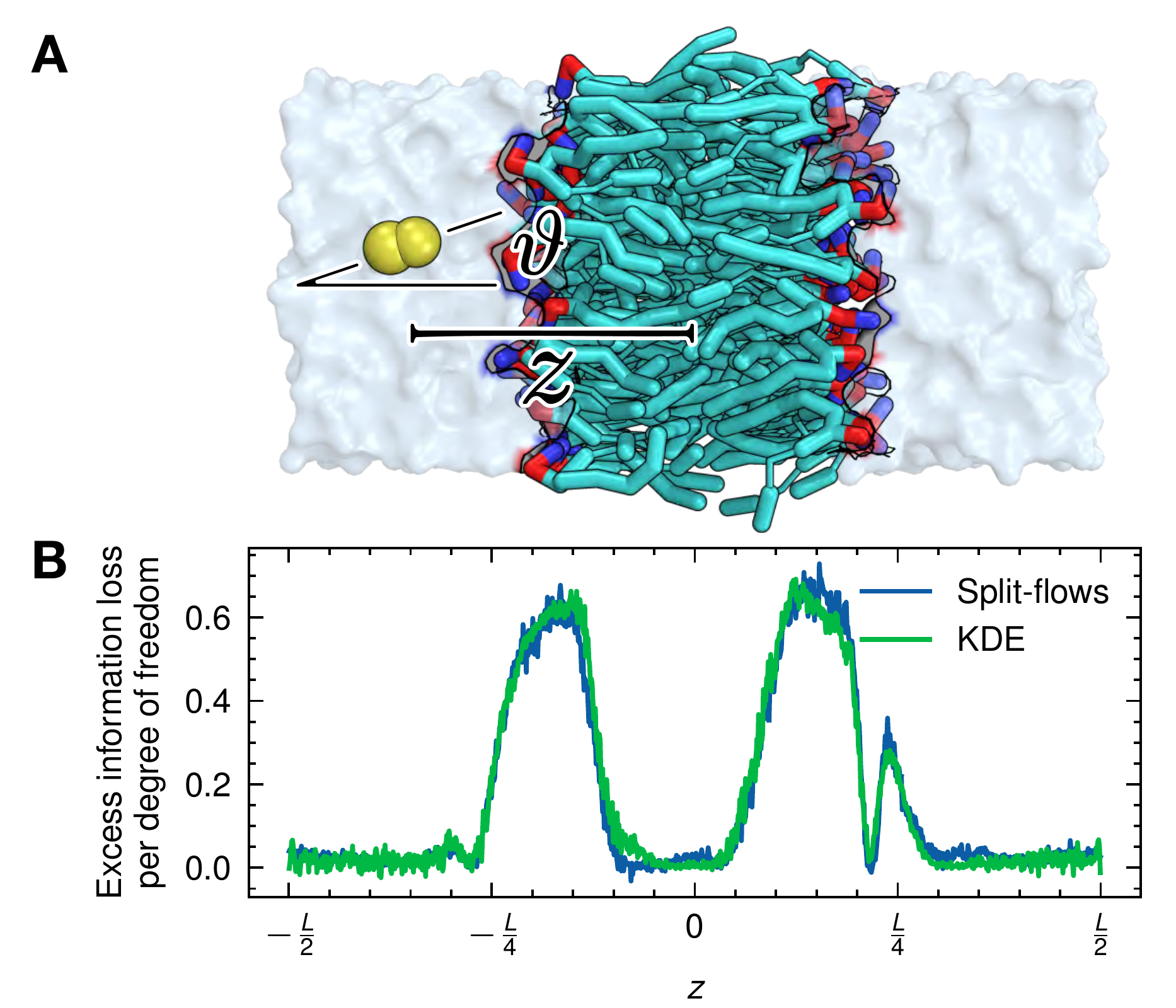}
    \caption{(A) An amphiphilic solute is dragged through a lipid bilayer surrounded by bulk water under a constant driving force. We describe its configuration by the distance $z$ to the membrane center and its relative orientation $\vartheta$ with respect to the $z$-axis.
    (B) Average excess information loss per degree of freedom when removing $\vartheta$, shown as a function of $z$.
    }
    \label{fig:lipid-bilayer}
\end{figure}

As a baseline, we estimate the fine- and coarse-grained densities using a simple kernel density estimator (KDE), yielding a binned approximation of the local information loss; see Appendix~\ref{app:bilayer} for details. Figure~\ref{fig:lipid-bilayer} shows the resulting excess information loss as a function of $z$. The split-flow estimates closely match the KDE baseline in both shape (Pearson correlation $0.99$) and local magnitude (mean absolute error $0.027$) across the coarse-grained domain.

The landscape reflects the amphiphilic interactions between the lipid membrane and the solute, and the associated constraints on the solute's relative orientation. In bulk water, these interactions are weak, and the solute's orientation is largely unconstrained, resulting in vanishing information loss. Near the surface, the hydrophilic headgroups attract the hydrophilic and repel the hydrophobic side of the solute, aligning it with the surface normal and causing a small peak in the information loss. Upon entering the membrane, the solute flips and orients its hydrophobic side toward the membrane interior. This re-orientation strongly constrains the solute, leading to a pronounced maximum in information loss at the interface. In the hydrophobic core, the constraint relaxes as both orientations become nearly equivalent, resulting in a clear decrease in information loss toward the bilayer midplane. This behavior is conceptually mirrored across the midplane, with a quantitative asymmetry due to the solute being pulled through the membrane with constant force.

\subsection{Alanine Dipeptide}
We consider 50k frames from a $1 \unit{\micro\second}$ MD simulation of alanine dipeptide at $600 \unit{\kelvin}$. We train the $E(3)$-equivariant GNN parameterization of split-flows, introduced in Section~\ref{sec:chignolin}, to connect the atomistic configuration with a reduced description in which only the five backbone atoms defining the dihedral angles $(\phi, \psi)$—the Ramachandran angles—are retained. This coarse-graining scheme is illustrated in Figure~\ref{fig:alanine-dipeptide-information-loss}. We provide a detailed description of the simulation and model hyperparameters in Appendix~\ref{app:ala2}. 

\begin{figure}[ht]
    \centering
    \includegraphics[width=1\linewidth]{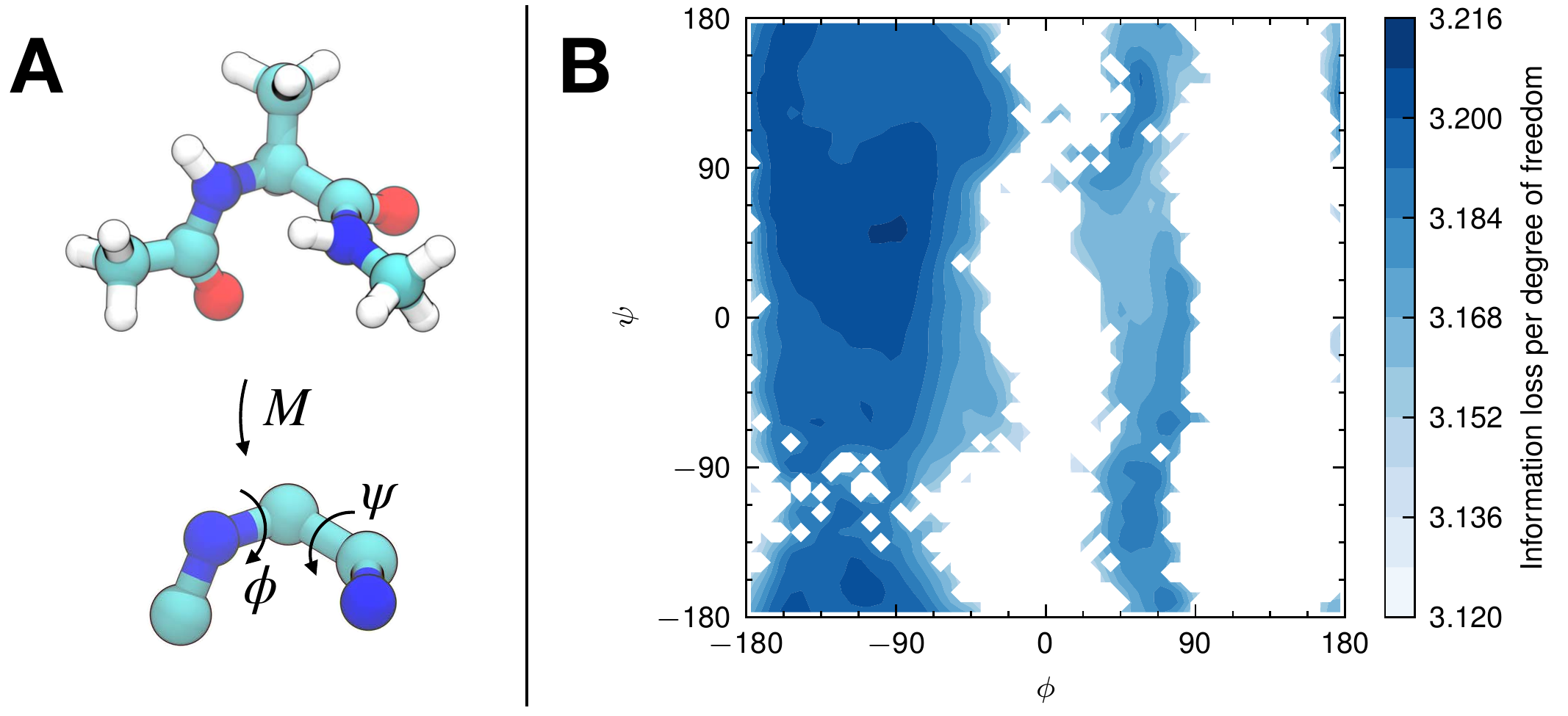}
    \caption{(A) The coarse-graining map reduces the atomistic configuration to the five atoms defining the Ramachandran dihedrals. (B) A landscape of the average information loss per removed degree of freedom is shown in the $(\phi, \psi)$-plane.}
    \label{fig:alanine-dipeptide-information-loss}
\end{figure}

Because the bond lengths and angles within this fragment are nearly rigid, the dihedrals $(\phi, \psi)$ uniquely determine the coarse-grained configuration $\mR$, up to global rigid-body motions corresponding to the Euclidean group $E(3)$. This enables us to visualize the information loss landscape over the two-dimensional $(\phi, \psi)$ plane in Figure~\ref{fig:alanine-dipeptide-information-loss}. This coarse-grained representation produces a complex distribution of lost information across the Ramachandran plane, reflecting the interactions of the eliminated degrees of freedom, including steric repulsions that generate the forbidden regions (white) and dipole–dipole interactions that shape the overall conformational preferences of the dipeptide. It demonstrates our model's ability to resolve highly non-trivial structure in the information loss of coarse-grained representations.

\section{CONCLUSION}
In this paper, we present split-flows, a novel approach for connecting molecular densities at different resolutions. Split-flows perform competitively in backmapping and—leveraging the volume change under the flow—can quantify the local information loss across general biophysical systems.

Our method performs well on various molecular systems. To scale up the method to larger macromolecules, autoregressive techniques---already used in the context of residue-based backmapping methods---may offer an appealing strategy. We propose to explore this direction in future work, where our contribution would provide unique insight in the scaling of resolution-based information loss.

Applications of our method are widespread. In particular, statistical thermodynamics provides a rich set of physical quantities linked to local mapping entropy, including the specific heat of a coarse configuration \cite{Foley2015}. In addition, split-flows offer a principled approach for identifying informative coarse-grained representations of molecular systems—specifically, those with low and uniform information loss. Finally, using split-flows to construct scale transitions in multi-scale molecular simulations represents a promising direction for future work.

\section{ACKNOWLEDGMENTS}
We thank Daniel Nagel for providing the lipid bilayer simulation and for fruitful discussions.
Furthermore, we gratefully acknowledge William Noid for providing extensive and constructive feedback on an earlier preprint of this work.
This work is supported by Deutsche Forschungsgemeinschaft (DFG) under Germany’s Excellence Strategy EXC-2181/1–390900948 (the Heidelberg STRUCTURES Excellence Cluster). 
The authors acknowledge support by the state of Baden-W\"urttemberg through bwHPC and the German Research Foundation (DFG) through grant INST 35/1597-1 FUGG.

\bibliography{literature}

\include{arXiv_CR_Supplement}

\end{document}

%% file: math_commands.tex
\usepackage{amsmath,amsfonts,bm}

\def\eqref#1{equation~\ref{#1}}

\def\1{\bm{1}}

\def\vh{{\bm{h}}}

\def\vm{{\bm{m}}}

\def\vr{{\bm{r}}}

\def\vt{{\bm{t}}}

\def\vv{{\bm{v}}}

\def\vx{{\bm{x}}}

\def\mR{{\bm{R}}}

\DeclareMathAlphabet{\mathsfit}{\encodingdefault}{\sfdefault}{m}{sl}
\SetMathAlphabet{\mathsfit}{bold}{\encodingdefault}{\sfdefault}{bx}{n}



%% file: arXiv_CR_Supplement.tex
\appendix
\clearpage
\thispagestyle{empty}

\onecolumn
\title{Split-Flows: Measure Transport and Information Loss Across Molecular Resolutions: \\
Supplementary Materials}

\section{ADDITIONAL PROOFS AND THEORETICAL DETAILS}
In this section, we provide additional theoretical details that complement the derivations of the local mapping entropy, information loss, and the split-flow setup in the main text. The propositions and proofs presented here rely heavily on well-established theory on normalizing flows \cite{rezende2015flows, Chen2018, Lipman2023, Albergo2023}. For completeness, we briefly state the core results.

Let $\phi_1: \mathbb{R}^n \rightarrow \mathbb{R}^n$ be a diffeomorphism and $\pi_0$ a probability density on $\mathbb{R}^n$. The density of the pushforward measure $(\phi_1)_\# \; \pi_0$ under the flow then can be written as
\begin{align}
    (\phi_1)_\# \, \pi_0(\vx_0) 
    = \left|\det J_{\phi_1}(\vx_0)\right|^{-1} \, \pi_0(\vx_0)
    = \pi_1(\phi_1(\vx_0)),
\end{align}
where $J_{\phi_1}(\vx_0)$ is the Jacobian matrix of partial derivatives. In case the flow $\phi_1$ is parameterized in continuous time by the ordinary differential equation (ODE):
\begin{equation}
    \tfrac{\mathrm{d}}{\mathrm{d} t} \phi_t(\vx_0) 
    =  \vv_t(\phi_t(\vx_0)), 
    \qquad \phi_0(\vx_0) = \vx_0,
\end{equation}
then the logarithm of the Jacobian determinant evolves according to the ODE:
\begin{equation}
    \tfrac{\mathrm{d}}{\mathrm{d} t} \log \left|\det J_{\phi_t}(\vx_0)\right|
    = \nabla \cdot \vv_t(\phi_t(\vx_0)),
\end{equation}
which yields an integral expression for the total change of volume along the flow:
\begin{equation}
    \log \left|\det J_{\phi_1}(\vx_0)\right|
    = \int_0^1 \mathrm{d}\tau \; \nabla \cdot \vv_\tau(\phi_\tau(\vx_0)).
\end{equation}
Consequently, the pushforward density can be expressed as
\begin{equation}
    \pi_1(\phi_1(\vx_0))
    = \pi_0(\vx_0) \,
    \exp \left[
      -\int_0^1 \mathrm{d}\tau \; \nabla \cdot \vv_\tau(\phi_\tau(\vx_0))
    \right].
\end{equation}

Together, these results formalize how probability densities evolve under smooth transformations and serve as the starting point for our theoretical analysis of mapping entropy, information loss, and split-flows.

\subsection{Decomposition of the Coarse-Grained Potential}
\label{app:pmf-decomposition}

We now show that the coarse-grained potential of mean force admits a natural decomposition into energetic and entropic contributions. \\

\begin{proposition}[Decomposition of the coarse-grained potential]
Let $\vr \in \mathbb{R}^n$ denote a fine-grained configuration, and let $\mR = M(\vr) \in \mathbb{R}^N$ be the associated coarse-grained representative obtained by a measurable coarse-graining map \( M: \mathbb{R}^n \to \mathbb{R}^N \). 
Suppose the fine-grained configurations are Boltzmann-distributed as
\begin{equation}
    \pi_r(\vr) = Z^{-1} \exp[-u(\vr)/(k_\mathrm{B} T)],
\end{equation}
where $u(\vr)$ is the potential energy governing the fine-grained distribution and $Z$ is the partition function. Then the marginal coarse-grained distribution $\pi_R(\mR)$ defined by
\begin{equation}
    \pi_R(\mR) = \int_{\mathbb{R}^n} \mathrm{d}\vr \, \delta(\mR - M(\vr)) \; \pi_r(\vr) 
\end{equation}
can be written in Boltzmann form:
\begin{equation}
    \pi_R(\mR) \propto \exp[-W(\mR)/(k_\mathrm{B} T)],
\end{equation}
where the free energy $W(\mR)$---the potential of mean force (PMF)---admits the decomposition
\begin{equation}\label{eq:app-pmf-decomposition}
    W(\mR) = E(\mR) - T S(\mR),
\end{equation}
with
\begin{equation}\label{eq:app-pmf-terms}
    E(\mR) = \mathbb{E}_{r \mid R}[u(\vr)],
    \qquad
    S(\mR) = -k_\mathrm{B} \, \mathbb{E}_{r \mid R}\!\left[\log \frac{\pi_r(\vr)}{\pi_R(\mR)}\right].
\end{equation}
\end{proposition}

\begin{proof}
Starting from the Boltzmann form of the coarse-grained density,
\begin{equation}
    \pi_R(\mR) \propto \exp[-W(\mR)/(k_\mathrm{B} T)],
\end{equation}
we express the PMF as
\begin{equation}
    W(\mR) = -k_\mathrm{B} T \log \pi_R(\mR) + \mathrm{const.}
    \label{eq:app-pmf-definition}
    \end{equation}
    Similarly, from the fine-grained Boltzmann distribution we can write
    \begin{equation}
    u(\vr) = -k_\mathrm{B} T \log \pi_r(\vr) + \mathrm{const.}
    \label{eq:app-fg-energy}
\end{equation}

Taking the conditional expectation of~\eqref{eq:app-fg-energy} over the conditional distribution 
$\pi_{r \mid R}$, we obtain
\begin{equation}
    \mathbb{E}_{r \mid R}[u(\vr)]
    = -k_\mathrm{B} T \, \mathbb{E}_{r \mid R}[\log \pi_r(\vr)] + \mathrm{const.}
    \label{eq:app-energy-exp}
\end{equation}

Inserting Equation~\ref{eq:app-energy-exp} into the terms of Equation~\ref{eq:app-pmf-decomposition}, written in  Equation~\ref{eq:app-pmf-terms}, we find:
\begin{align}
    W(\mR) &= E(\mR) - T S(\mR) \\ 
    &= \mathbb{E}_{r \mid R} [u(\vr)] + k_\mathrm{B} T \; \mathbb{E}_{r \mid R} \left [ \log \frac{\pi_r(\vr)}{\pi_R(\mR)} \right ] \\ 
    &= \mathbb{E}_{r \mid R} [u(\vr)] - \mathbb{E}_{r \mid R} [u(\vr)] - k_\mathrm{B} T \; \mathbb{E}_{r \mid R} [\log \pi_R(\mR)] + \mathrm{const.} \\ 
    &= - k_\mathrm{B} T \; \mathbb{E}_{r \mid R} [\log \pi_R(\mR)] + \mathrm{const.},
\end{align}
which recovers Equation~\ref{eq:app-pmf-definition} up to the unknown constant offset. However, since potentials are defined only up to an additive constant, we can drop the constant, completing the proof.
\end{proof}

\begin{remark}
The decomposition \(W = E - TS\) expresses the PMF as the conditional free energy of the fine-grained system constrained to the coarse-grained configuration \(\mR\). Here, \(E(\mR)\) denotes the mean internal energy of the fine-grained microstates compatible with \(\mR\), while \(S(\mR)\) quantifies their configurational entropy.

\end{remark}

\begin{remark}
The entropic contribution to the PMF, $S(\mR)$, can be identified with the local mapping entropy $S(\mR)$, which quantifies the information loss in a coarse-grained representation. Higher information loss corresponds to a lower mapping entropy and thus a higher value of the PMF, which in turn lowers the probability of the coarse-grained configuration.
\end{remark}

\subsection{Computation of Fiber Averages with Split-Flows} \label{app:fiber-averages}
Split flows allow us to directly access fiber averages, i.e., expectation values of observables defined on the fine-grained space \( \mathbb{R}^n \), restricted to the fiber \( \mathcal{F}(\mR) \) associated with a coarse-grained representative \( \mR \). We formalize the expression stated in Equation~\ref{eq:fiber-average-split-flow} in the main text with the following proposition.
\\

\begin{proposition}[Computation of fiber averages with split-flows] \label{prop:fiber-averages}
Let $\pi_R$ be a probability density on $\mathbb{R}^N$ and $\pi_{\epsilon \mid R}$ a conditional density on $\mathbb{R}^{n-N}$, defining a joint density $\pi_R \times \pi_{\epsilon \mid R}$ on $\mathbb{R}^n = \mathbb{R}^N \times \mathbb{R}^{n-N}$. Let $M : \mathbb{R}^n \to \mathbb{R}^N$ be a measurable coarse-graining map, and let $\phi_1 : \mathbb{R}^n \to \mathbb{R}^n$ be a diffeomorphism satisfying
\begin{equation} \label{eq:app-pushforward-fiber-average}
    (\phi_1)_\# \; \pi_R(\mR) \pi_{\epsilon \mid R}(\boldsymbol{\epsilon}) 
    = \left|\det J_{\phi_1}(\mR, \boldsymbol{\epsilon})\right|^{-1} \, \pi_R(\mR)\, \pi_{\epsilon \mid R} (\boldsymbol{\epsilon}\mid \mR ) 
    = \pi_r(\phi_1(\mR, \boldsymbol{\epsilon})),
\end{equation}
where $\pi_r$ is a target density on $\mathbb{R}^n$. Assume further that $\phi_1$ inverts the coarse-graining map in the sense that
\begin{equation}
    M \circ \phi_1(\mR, \boldsymbol{\epsilon}) = \mR,
\end{equation}
for all $(\mR, \boldsymbol{\epsilon}) \in \mathbb{R}^N \times \mathbb{R}^{n-N}$. Finally, let $O : \mathbb{R}^n \to \mathbb{R}^d$ be a measurable observable. We can then write the conditional expectation of $O$ over $\pi_{r \mid R}$ as:
\begin{equation}
    \mathbb{E}_{r \mid R}[O(\vr)] = \mathbb{E}_{\epsilon \mid R}[O(\phi_1(\mR, \boldsymbol{\epsilon}))].
\end{equation}
\end{proposition}

\begin{proof}
Let $\pi_r$, $\pi_R$, $\pi_{\epsilon \mid R}$, $\phi_1$, $M$ and $O$ obey the properties above. By definition, we can write the fiber average of $O$ for a given coarse-grained representative $\mR$ as
\begin{equation}
    \mathbb{E}_{r \mid R}[O(\vr)]
    = \frac{\int_{\mathbb{R}^n} \mathrm{d}\vr \, \delta(M(\vr) - \mR) \; O(\vr)\, \pi_r(\vr)}
    {\int_{\mathbb{R}^n} \mathrm{d}\vr \, \delta(M(\vr) - \mR) \; \pi_r(\vr)}.
\end{equation}
By substituting $\vr=\phi_1(\mR', \boldsymbol{\epsilon})$ and using the change-of-variables theorem
$d\vr = \left|\det J_{\phi_1}(\mR',\boldsymbol{\epsilon})\right|\, d\mR'\, d\boldsymbol{\epsilon}$, we can rewrite the expectation as
\begin{align}
    \mathbb{E}_{r \mid R}[O(\vr)]
    &= \frac{\int_{\mathbb{R}^n} \mathrm{d}\vr \, \delta(M(\vr) - \mR) \; O(\vr) \pi_r(\vr)}
    {\int_{\mathbb{R}^n} \mathrm{d}\vr \, \delta(M(\vr) - \mR) \; \pi_r(\vr)} \\
    &= \frac{\int_{\mathbb{R}^N}\!\!\int_{\mathbb{R}^{n-N}}  
    \mathrm{d}\mR' \mathrm{d} \boldsymbol{\epsilon} \,
    \delta(M(\phi_1(\mR', \boldsymbol{\epsilon})) - \mR) \;
    O(\phi_1(\mR', \boldsymbol{\epsilon}))\,
    \pi_r(\phi_1(\mR', \boldsymbol{\epsilon}))\,
    \left|\det J_{\phi_1}(\mR', \boldsymbol{\epsilon})\right|}
    {\int_{\mathbb{R}^N} \!\!\int_{\mathbb{R}^{n-N}}
    \mathrm{d}\mR' \mathrm{d} \boldsymbol{\epsilon} \,
    \delta(M(\phi_1(\mR', \boldsymbol{\epsilon})) - \mR) \;
    \pi_r(\phi_1(\mR', \boldsymbol{\epsilon}))\,
    \left|\det J_{\phi_1}(\mR', \boldsymbol{\epsilon})\right|} \\
    &\overset{M\circ\phi_1=\mathrm{Id}_R}{=}
    \frac{\int_{\mathbb{R}^{n-N}}
    \mathrm{d}\boldsymbol{\epsilon} \;
    O(\phi_1(\mR, \boldsymbol{\epsilon}))\,
    \pi_r(\phi_1(\mR, \boldsymbol{\epsilon}))\,
    \left|\det J_{\phi_1}(\mR, \boldsymbol{\epsilon})\right|}
    {\int_{\mathbb{R}^{n-N}}
    \mathrm{d}\boldsymbol{\epsilon} \;
    \pi_r(\phi_1(\mR, \boldsymbol{\epsilon}))\,
    \left|\det J_{\phi_1}(\mR, \boldsymbol{\epsilon})\right|},
\end{align}
where $J_{\phi_1}(\mR, \boldsymbol{\epsilon})$ denotes the Jacobian matrix of partial derivatives of $\phi_1$. Identifying
\[
\pi_r\ (\phi_1(\mR, \boldsymbol{\epsilon}) )\,
\left|\det J_{\phi_1}(\mR, \boldsymbol{\epsilon})\right|
= \pi_R(\mR)\, \pi_{\epsilon \mid R} (\boldsymbol{\epsilon}\mid \mR )
\]
via the pushforward relation in Equation~\ref{eq:app-pushforward-fiber-average}, we obtain
\begin{align}
    \mathbb{E}_{r \mid R}[O(\vr)] 
    &= \frac{\int_{\mathbb{R}^{n-N}} 
    \mathrm{d}\boldsymbol{\epsilon}\; 
    O (\phi_1(\mR, \boldsymbol{\epsilon}))\,
    \pi_R(\mR)\, 
    \pi_{\epsilon \mid R}(\boldsymbol{\epsilon}\mid \mR)}
    {\int_{\mathbb{R}^{n-N}} 
    \mathrm{d}\boldsymbol{\epsilon}\; 
    \pi_R(\mR)\, 
    \pi_{\epsilon \mid R}(\boldsymbol{\epsilon}\mid \mR)} \\
    &= \int_{\mathbb{R}^{n-N}} 
    \mathrm{d}\boldsymbol{\epsilon}\; 
    O(\phi_1(\mR, \boldsymbol{\epsilon}))\,
    \pi_{\epsilon \mid R}(\boldsymbol{\epsilon}\mid \mR)\\
    &= \mathbb{E}_{\epsilon \mid R}[O(\phi_1(\mR, \boldsymbol{\epsilon}))],
\end{align}
which completes the proof.
\end{proof}

\begin{remark}[Practical estimation] \label{remark:mc-estimate}
In practice, we approximate the expectation value over $\pi_{\epsilon \mid \mR}$ using a Monte Carlo estimate:
\begin{equation}
    \mathbb{E}_{\epsilon \mid R}[O(\phi_1(\mR, \boldsymbol{\epsilon}))] \approx \frac{1}{M} \sum_{i=1}^M O(\phi_1(\mR, \boldsymbol{\epsilon}_i)).
\end{equation}
Here, $\boldsymbol{\epsilon}_i$ are $M$ samples drawn from the noise distribution $\pi_{\epsilon \mid \mR}$, which is straightforward to sample from by construction.
\end{remark}

\subsection{Mapping Entropy Estimation with Split-Flows}
Split-flows allow us to obtain an unbiased estimate of the local mapping entropy $S(\mR)$ for a given coarse-grained configuration $\mR$. This section will complement the derivations of Equation~\ref{eq:local-mapping-entropy-split-flow} in the main text with a formal proposition, showing that the estimator arises naturally from the pushforward structure of the flow. \\

\begin{proposition}[Mapping entropy estimation with split-flows]
Let $\pi_R$ be a probability density on $\mathbb{R}^N$ and $\pi_{\epsilon \mid R}$ a conditional density on $\mathbb{R}^{n-N}$, defining a joint density $\pi_R \times \pi_{\epsilon \mid R}$ on $\mathbb{R}^n = \mathbb{R}^N \times \mathbb{R}^{n-N}$.  Let $M : \mathbb{R}^n \to \mathbb{R}^N$ be a measurable coarse-graining map, and let $\phi_1 : \mathbb{R}^n \to \mathbb{R}^n$ be a diffeomorphism satisfying
\begin{equation} \label{eq:app-pushforward-me}
    (\phi_1)_\# \; \pi_R(\mR) \pi_{\epsilon \mid R}(\boldsymbol{\epsilon}) 
    = \left|\det J_{\phi_1}(\mR, \boldsymbol{\epsilon})\right|^{-1} \pi_R(\mR)\, \pi_{\epsilon \mid R} (\boldsymbol{\epsilon}\mid \mR) 
    = \pi_r(\phi_1(\mR, \boldsymbol{\epsilon})),
\end{equation}
where $\pi_r$ is a target density on $\mathbb{R}^n$. Assume further that $\phi_1$ inverts the coarse-graining map in the sense that
\begin{equation}
    M \circ \phi_1(\mR, \boldsymbol{\epsilon}) = \mR,
\end{equation}
for all $(\mR, \boldsymbol{\epsilon}) \in \mathbb{R}^N \times \mathbb{R}^{n-N}$. The local mapping entropy, defined as:
\begin{equation}
    S(\mR) = - k_\mathrm{B} \int_{\mathcal{F}(\mR)} \mathrm{d}\vr \; \pi_{r \mid R}(\vr \mid \mR) \log \pi_{r \mid R}(\vr \mid \mR)
\end{equation}
then can be estimated via the split-flow setup as:
\begin{equation}
    S (\mR)
    = - k_\mathrm{B} \, \mathbb{E}_{\epsilon \mid R} \left[
    \log \pi_{\epsilon \mid R} (\boldsymbol{\epsilon}\mid \mR)
    \right ]
    + k_\mathrm{B} \, \mathbb{E}_{\epsilon \mid R} \left[ 
    \log  \left|\det J_{\phi_1}(\mR,\boldsymbol{\epsilon})\right|
    \right].
\end{equation}
\end{proposition}

\begin{proof}
Let $\phi_1$, $\pi_r$, and $\pi_R$ fulfill the properties above. Starting from the definition of the local mapping entropy, we use Bayes' formula to rewrite the integrand:
\begin{align}
    S(\mR)
    &= - k_\mathrm{B} \int_{\mathcal{F}(\mR)} \mathrm{d}\vr \; \pi_{r \mid R}(\vr \mid \mR) \log \pi_{r \mid R}(\vr \mid \mR) \\
    &= - k_\mathrm{B} \int_{\mathcal{F}(\mR)} \mathrm{d}\vr \; \pi_{r \mid R}(\vr \mid \mR) \log \frac{\pi_{R \mid r}(\mR \mid \vr) \pi_r(\vr)}{\pi_R(\mR)} \\
    &= - k_\mathrm{B} \int_{\mathcal{F}(\mR)} \mathrm{d}\vr \; \pi_{r \mid R}(\vr \mid \mR) \log \frac{\pi_r(\vr)}{\pi_R(M(\vr))} \\
    &= - k_\mathrm{B} \, \mathbb{E}_{r \mid R} \left[ \log \frac{\pi_r(\vr)}{\pi_R(M(\vr))} \right]
\end{align}
where we used that the posterior $\pi_{R \mid r}(\mR \mid \vr) \equiv 1$ on the integration domain $\mathcal{F}({\mR})$ and replaced $\mR$ by $M(\vr)$. Next we are going to leverage the result of Proposition~\ref{prop:fiber-averages} to obtain:
\begin{align}
    S(\mR)
    &= - k_\mathrm{B} \, \mathbb{E}_{r \mid R} \left[ \log \frac{\pi_r(\vr)}{\pi_R(M(\vr))} \right] \\
    &= - k_\mathrm{B} \, \mathbb{E}_{\epsilon \mid R} \left[ \log \frac{\pi_r(\phi_1(\mR, \boldsymbol{\epsilon}))}{\pi_R(M(\phi_1(\mR, \boldsymbol{\epsilon})))} \right] \\
    &= - k_\mathrm{B} \, \mathbb{E}_{\epsilon \mid R} \left[ \log \frac{\pi_r(\phi_1(\mR, \boldsymbol{\epsilon}))}{\pi_R(\mR)} \right].
\end{align}
Inserting the pushforward relation in Equation~\ref{eq:app-pushforward-me} then yields:
\begin{align}
    S(\mR)
    &= - k_\mathrm{B} \, \mathbb{E}_{\epsilon \mid R} \left[ \log \frac{\pi_r(\phi_1(\mR, \boldsymbol{\epsilon}))}{\pi_R(\mR)} \right] \\
    &= - k_\mathrm{B} \, \mathbb{E}_{\epsilon \mid R} \left[
    \log \pi_{\epsilon \mid R} (\boldsymbol{\epsilon}\mid \mR)
    \right ]
    + k_\mathrm{B} \, \mathbb{E}_{\epsilon \mid R} \left[ 
    \log  \left|\det J_{\phi_1}(\mR,\boldsymbol{\epsilon})\right|
    \right],
\end{align}
which completes the proof.
\end{proof}

\begin{remark}[Mapping entropy estimation with continuous normalizing flows]
In case the flow $\phi_t$ is parameterized in continuous time by an underlying velocity field $\vv_t$, we can replace its log Jacobian determinant with an integral over the interpolation interval $[0, 1]$ of the divergence of the flow \cite{Lipman2023}:
\begin{equation}
\log \left|\det J_{\phi_1}(\mR,\boldsymbol{\epsilon})\right|
= \int_0^1 \mathrm{d}\tau \; \nabla \cdot \vv_\tau (\phi_\tau(\mR, \boldsymbol{\epsilon})).
\end{equation}
This recovers the estimator presented in Equation~\ref{eq:local-mapping-entropy-split-flow}.
\end{remark}

\begin{remark}[Variance of the mapping entropy estimator]
In practice, we approximate the expectation over $\pi_{\epsilon \mid R}$ in the estimate of the local mapping entropy in Equation~\ref{eq:local-mapping-entropy-split-flow} using the Monte Carlo estimator presented in Remark~\ref{remark:mc-estimate}. In Figure~\ref{fig:app-mapping-entropy-variance}, we analyze the variance of this estimator for the solute-in-lipid-bilayer experiment described in Section~\ref{sec:lipid-bilayer}. We find that the relative variance of the estimator follows a power-law decay with respect to the number of samples.
\begin{figure}[ht]
    \centering
    \includegraphics[width=0.5\linewidth]{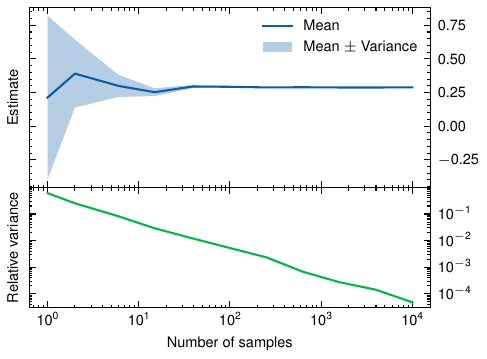}
    \caption{Estimate and variance of the local excess mapping entropy as a function of the number of samples. We show the estimate (top) and the relative variance with respect to the mean (bottom). The linear behavior of the relative variance in the log-log plot indicates a power-law decay.}
    \label{fig:app-mapping-entropy-variance}
\end{figure}
\end{remark}

\subsection{Validity of the Split-Flow Coupling}
The training algorithm for two-sided flow matching relies on drawing samples from the endpoint distributions. In the split-flow setup, we propose constructing a coupling based on the coarse-graining map, which encourages the flow to correctly pair fine- and coarse-grained configurations. We show that this coupling is a valid coupling. \\

\begin{proposition}[Validity of the split-flow coupling]
Let $\pi_r$ be a probability density on $\mathbb{R}^n$. Let $\pi_R$ be a probability density on $\mathbb{R}^N$ and let $\pi_{\epsilon \mid R}$ be a conditional density on $\mathbb{R}^{n-N}$, defining a joint density $\pi_R \times \pi_{\epsilon \mid R}$ on $\mathbb{R}^n = \mathbb{R}^N \times \mathbb{R}^{n-N}$. Finally, let $M : \mathbb{R}^n \to \mathbb{R}^N$ be a coarse-graining map. The joint coupling
\begin{equation}
    \pi_{R, \epsilon, r}(\mR, \boldsymbol{\epsilon}, \vr)
    = \pi_r(\vr) \, \delta(\mR - M(\vr)) \, \pi_{\epsilon \mid R}(\boldsymbol{\epsilon} \mid \mR).
\end{equation}
then defines a valid coupling of the two distributions $\pi_r$ and $\pi_R \times \pi_{\epsilon \mid R}$ in the sense that the marginal over $\vr$ is $\pi_r$, and the marginal over $(\mR, \boldsymbol{\epsilon})$ is $\pi_R \times \pi_{\epsilon \mid R}$.
\end{proposition}

\begin{proof}
Let $\pi_r$, $\pi_R$, $\pi_{\epsilon \mid R}$, and $M$ be defined as above. The marginal of $\pi_{R, \epsilon, r}$ over $\vr$ then reads:
\begin{align}
    \int_{\mathbb{R}^N}\mathrm{d}\mR 
    \int_{\mathbb{R}^{n-N}}\mathrm{d}\boldsymbol{\epsilon} \;
    \pi_{R, \epsilon, r}(\mR, \boldsymbol{\epsilon}, \vr) 
    = \int_{\mathbb{R}^N}\mathrm{d}\mR 
    \int_{\mathbb{R}^{n-N}}\mathrm{d}\boldsymbol{\epsilon} \;
    \pi_r(\vr) \, \delta(\mR - M(\vr)) \, \pi_{\epsilon \mid R}(\boldsymbol{\epsilon} \mid \mR) 
    = \pi_r(\vr),
\end{align}
since the integral over $\mR$ evaluates at $\mR = M(\vr)$ and the inner integral over $\boldsymbol{\epsilon}$ simply integrates to $1$ due to normalization of $\pi_{\epsilon \mid R}$. Furthermore, the marginal over $(\mR, \boldsymbol{\epsilon})$ reads:
\begin{equation} \label{eq:app-marginal-cg}
    \int_{\mathbb{R}^n}\mathrm{d}\vr \;     
    \pi_{R, \epsilon, r}(\mR, \boldsymbol{\epsilon}, \vr) 
    = \underbrace{\int_{\mathbb{R}^n}\mathrm{d}\vr \;
    \pi_r(\vr) \, \delta(\mR - M(\vr))}_{= \, \pi_R(\mR)} \, \pi_{\epsilon \mid R}(\boldsymbol{\epsilon} \mid \mR) 
    = \pi_R(\mR) \, \pi_{\epsilon \mid R}(\boldsymbol{\epsilon} \mid \mR),
\end{equation}
where we used the definition of the coarse grained density in Equation~\ref{eq:explicit-marginal}.
\end{proof} 

\begin{remark}[Dimensionality bridging via augmentation]  
By augmenting the coarse-grained configuration \( \mR \in \mathbb{R}^N \) with auxiliary noise \( \boldsymbol{\epsilon} \in \mathbb{R}^{n - N} \), we construct a joint variable \( (\mR, \boldsymbol{\epsilon}) \in \mathbb{R}^n \) that enables training a continuous normalizing flow \( \phi_t : \mathbb{R}^n \rightarrow \mathbb{R}^n \) via flow matching. Since both endpoint distributions now are defined over the same space, the flow is well-defined and can learn a bijective transport from \( \pi_R \times \pi_{\epsilon \mid R} \) to \( \pi_r \). This approach resolves the non-invertibility of the coarse-graining map by converting the ill-posed inverse problem into a generative one.
\end{remark}

\subsection{A Geometric Perspective} \label{app:geometric-perspective}

\begin{figure}[ht]
    \centering
    \includegraphics[width=0.7\linewidth]{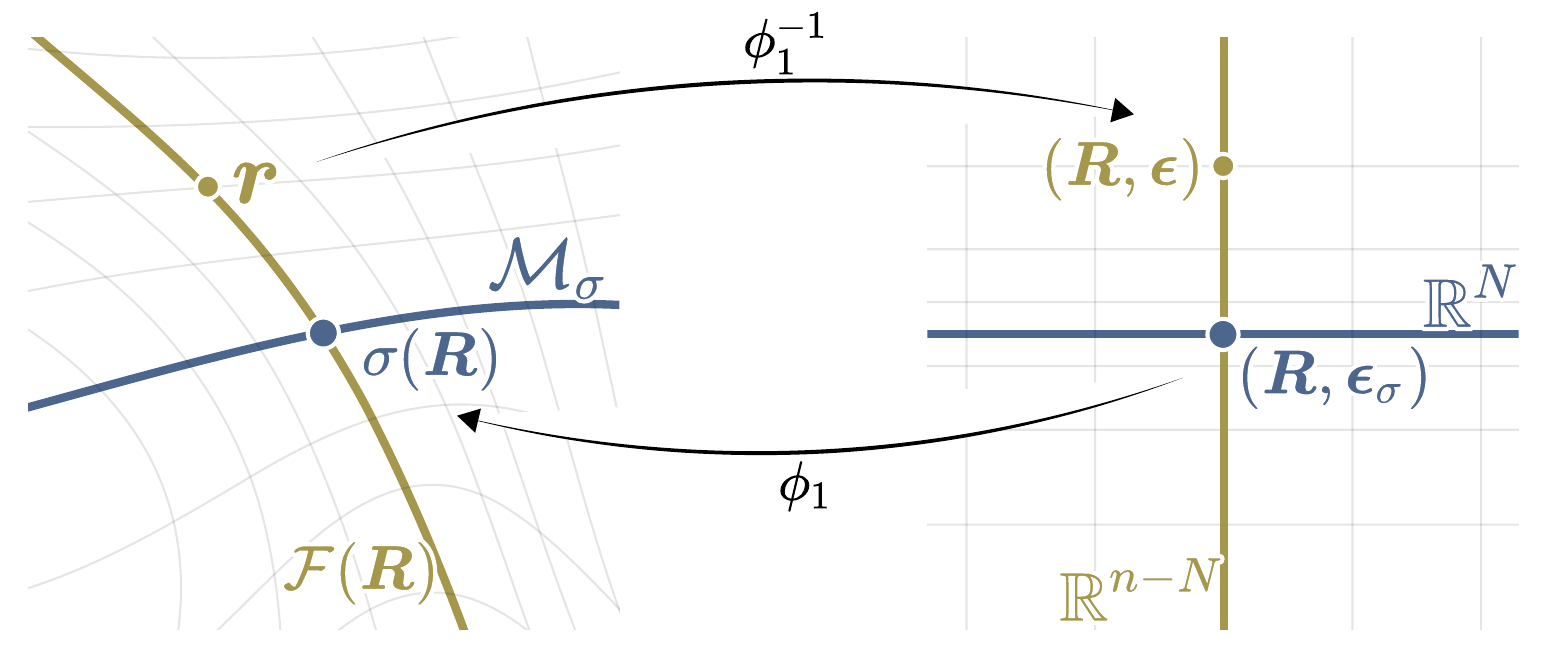}
    \caption{A geometric perspective on split-flows. The coarse-graining map $M$ uniquely defines the fiber $\mathcal{F}(\boldsymbol{R})$ over a coarse-grained representation $\boldsymbol{R}$. By fixing a section $\sigma(\boldsymbol{R})$, we can define a coarse manifold $\mathcal{M}_\sigma$. Split-flows then learn a global coordinate transformation $\vr \mapsto (\mR, \boldsymbol{\epsilon})$, that disentangles the geometric structure induced by the coarse-graining map.}
    \label{fig:geometric-perspective}
\end{figure}

In this section we provide a geometric view on split-flows (see Figure~\ref{fig:geometric-perspective}), following ideas introduced by \cite{Brehmer2020}. The underlying principle is to disambiguate between (i) manifold-learning, i.e., identifying the underlying manifold $\mathcal{M}_\sigma$ of slow degrees of freedom and the orthogonal fiber $\mathcal{F}(\mR)$ of fast degrees of freedom, and (ii) estimating the conditional density $\pi_{r \mid R}$ of the fast degrees of freedom on the fiber.

\textit{Fiber:} We consider a smooth coarse-graining map $M: \mathbb{R}^n \rightarrow \mathbb{R}^N$, which associates fine-grained configurations $\vr \in \mathbb{R}^n$ to a coarse-grained representative $M(\vr) = \mR \in \mathbb{R}^N$. The fiber $\mathcal{F}(\mR)$ over a coarse-grained representative $\mR$ is the set
\begin{equation}
    \mathcal{F}(\mR) = \{\vr \in \mathbb{R}^n \mid M(\vr) = \mR\}.
\end{equation}
Its position and shape are uniquely determined by the position $\mR$ and the mapping $M$. Moving along the fiber changes the fast degrees of freedom, while keeping the slow degrees of freedom fixed. Assuming that the coarse-graining map $M$ is a smooth submersion, i.e., its Jacobian $J_M$ has full rank, the fibers for different values of $\mR$ form a foliation of the configurational space $\mathbb{R}^n$.

\textit{Local tangent space decomposition:} Due to the submersion theorem \cite{Lee2003}, we can locally decompose the tangent space $T_{\vr}\mathbb{R}^n$ into the fast (fiber) and slow (coarse) directions induced by $M$ at $\vr \in \mathbb{R}^n$:
\begin{equation}
    T_{\vr}\mathbb{R}^n = \ker J_M(\vr) \,  \oplus \, \mathrm{range} \, J_M(\vr)^\top,
\end{equation}
where $J_M(\vr) \in \mathbb{R}^{N \times n}$ is the Jacobian matrix of partial derivatives of $M$, which we assume to have full rank. We can further identify the two components as local tangent spaces of the fiber $\mathcal{F}(\mR)$ and a locally defined coarse manifold $\mathcal{M}$, respectively:
\begin{equation}
    T_{\vr}\mathcal{F}(\mR) = \ker J_M(\vr), \qquad T_{\vr} \mathcal{M} = \mathrm{range} \, J_M(\vr)^\top.
\end{equation}

\textit{Gauge-freedom of the coarse manifold:} While the fiber is uniquely defined, the definition of a coarse manifold $\mathcal{M}$ is ambiguous. To resolve this, we choose a section $\sigma(\mR) \in \mathcal{F}(\mR)$ to uniquely define a coarse manifold $\mathcal{M}_\sigma$, e.g., \cite{Pimentel2025a} propose to use the minimum-energy configuration on the fiber. In general, the choice of section $\sigma: \mathbb{R}^N \rightarrow \mathbb{R}^n$, remains a gauge-freedom, which must be fixed to select a unique coarse manifold among the family of manifolds compatible with the local tangent-space decomposition induced by $M$.

\textit{Globally linearized structure:} Split-flows parametrize a diffeomorphic coordinate map 
\begin{equation}
    \phi_1^{-1}: \mathbb{R}^n \to \mathbb{R}^N \times \mathbb{R}^{n-N}, \qquad \vr \mapsto \left( \mR, \boldsymbol{\epsilon} \right)
\end{equation}
whose first $N$ components coincide with the coarse-graining map $M$. In these coordinates, the fibers $\mathcal{F}(\mR)$ are mapped to affine
subspaces $\{\mR\}\times\mathbb{R}^{n-N}$, while the coarse manifold
$\mathcal{M}_\sigma$ is mapped to the coordinate subspace
$\mathbb{R}^N\times\{\boldsymbol{\epsilon}_\sigma\}$. Thereby, split flows (i) learn a coordinate system that disentangles the geometric structure induced by the coarse-graining map $M$ into slow (coarse) and fast (fiber) degrees of freedom, and (ii) enable density estimation on each fiber by mapping the corresponding conditional distribution $\pi_{r \mid R}$ on $\mathcal{F}(\mR)$ to a tractable density $\pi_{\epsilon \mid R}$ on the affine subspace $\{\mR\}\times\mathbb{R}^{n-N}$.

\section{EXPERIMENTAL DETAILS} \label{app:experimental-details}

In this section we are going to provide additional details on the experiments performed in the experimental section of the main text. These include details on data generation, model parameterization, evaluation, and model training.

\subsection{Chignolin} \label{app:chignolin}

\subsubsection{Data Generation} \label{app:chignolin-data-generation}
\textit{Langevin molecular dynamics:} Molecular dynamics (MD) is a trajectory-based sampling algorithm that simulates the time evolution of the configuration $\mathbf{r}$ of a molecular system using Newton's equations of motion. MD simulations are often performed at constant temperature $T$, which is maintained by coupling the system to an external heat bath—a procedure referred to as \emph{thermostatting}. One widely used approach is the Langevin equation, which augments Newtonian dynamics with friction and stochastic thermal forces:
\begin{equation} \label{eq:app-langevin-dynamics}
    m_i \frac{\mathrm{d}^2 \vr_i}{\mathrm{d}t^2}
    = - \nabla_{\vr_i} u(\vr)
    - \gamma m_i \frac{\mathrm{d}\vr_i}{\mathrm{d}t}
    + \sqrt{2 \gamma m_i k_\mathrm{B} T}\, \boldsymbol{\zeta}_i(t),
\end{equation}
where $m_i$ and $\vr_i$ denote the mass and position of the $i$-th particle, $u(\vr)$ is the potential energy function, $\gamma$ is the friction coefficient, $k_\mathrm{B}$ is Boltzmann's constant, and $\boldsymbol{\zeta}_i(t)$ is Gaussian white noise with zero mean and unit variance. The deterministic term $-\nabla_{\vr_i} u(\vr)$ drives the system according to interatomic forces, the friction term dissipates kinetic energy, and the stochastic term restores energy from the thermal bath. Together, these ensure that the equilibrium distribution of sampled configurations $\vr$ is given by:
\begin{equation}
    \pi_r(\vr)
    \propto \exp \left[-u(\vr) / (k_\mathrm{B} T) \right].
\end{equation}
In practice Equation~\ref{eq:app-langevin-dynamics} is discretized with a finite timestep $\Delta t$.

\textit{Simulation:} 
To obtain training data, we simulate the mini-protein chignolin using Langevin molecular dynamics in OpenMM with the AMBER14 all-atom force field and the TIP3P water model. The chignolin peptide (PDB ID: 1UAO) is solvated in a cubic water box with $1 \unit{\nano\meter}$ padding and neutralized. Simulations are performed at $360 \unit{\kelvin}$ using Langevin dynamics ($1 \unit{\per\pico\second}$ friction, $2 \unit{\femto\second}$ timestep) with PME electrostatics and a $1 \unit{\nano\meter}$ cutoff. After energy minimization and velocity initialization, we simulate the system for $1 \unit{\micro\second}$ and save coordinates every $2 \unit{\pico\second}$.

\textit{Data preparation:}
After simulation, we apply preprocessing steps to the raw data. These include removing the water solvent and hydrogen atoms, retaining only the heavy atoms of the protein. We then center the coordinates of each individual configuration and superpose the configurations along the simulated trajectory.

\subsubsection{Model Parameterization} \label{app:chignolin-model-parameterization}
\textit{Network architecture:} To parameterize the flow, we use the $E(3)$-equivariant graph neural network (GNN) architecture proposed by \cite{Satorras2021}. Equivariance is achieved using equivariant graph convolutional layers (EGCL). Given coordinates $\vx_i^{(l)} \in \mathbb{R}^3$ and node embeddings $\vh_i^{(l)} \in \mathbb{R}^{d_H}$ for each node $i$, the output node features and coordinates of the $l$-th EGCL layer are computed as follows:
\begin{align}
    \vm_{ij} &= \varphi_e \left(\vh_i^{(l)}, \vh_j^{(l)}, \gamma \left(\|\vx_i^{(l)} - \vx_j^{(l)}\|^2\right), a_{ij} \right),\\
    \vx_i^{(l+1)} &= \vx_i^{(l)} + \frac{1}{M_i-1} \sum_{j \neq i} \left (\vx_i^{(l)} - \vx_j^{(l)} \right )\, \varphi_x(\vm_{ij}),\\
    \vm_i &= \sum_{j \neq i} \vm_{ij},\\
    \vh_i^{(l+1)} &= \varphi_h\left(\vh_i^{(l)}, \vm_i^{(l)}\right).
\end{align}
Here, $\varphi$ denotes functions parameterized by multi-layer perceptrons (MLPs), and $\gamma$ represents a $d_F$-dimensional Fourier feature encoding function of the distance between two node coordinates. Furthermore, $a_{ij}$ denotes information associated with the edge between nodes $i$ and $j$, and $M_i$ denotes the number of nodes in the one-hop neighborhood of node $i$. The full network consists of $L$ such layers. As initial hidden node embeddings $\vh_i^{(0)}$, we use a concatenation of linear embeddings of the particle’s atom type $a \in [0, 1]^{N_A}$, associated bead type $b \in [0, 1]^{N_B}$, and the interpolation time $t \in [0, 1]$. Moreover, we do not include additional edge information $a_{ij}$.

\textit{Noise distribution:} For the projected-out atoms, we define a residue-wise target latent distribution. The latent position $\boldsymbol{\epsilon}_{I, i}$ of the $i$-th atom in the $I$-th residue is sampled from a Gaussian distribution centered at the position $\mR_I$ of the corresponding $C_\alpha$ atom:
\begin{equation}
\boldsymbol{\epsilon}_{I,i} \sim \mathcal{N}(\mR_I,\sigma^2 \mathbf{1}),
\end{equation}
with variance $\sigma^2$. 

We report the hyperparameter choices for the model's architecture and for the noise distribution in Table~\ref{tab:app-chignolin-benchmark-architecture}.

\begin{table}[ht]
\caption{Architectural hyperparameters of the model trained for backmapping the $C_\alpha$ representation of chignolin. We report the choices for the $E(3)$-equivariant GNN parameterization of the velocity field and the noise distribution.}  
\label{tab:app-chignolin-benchmark-architecture}
\centering
\begin{tabular}{lc}
\textbf{Hyperparameter} & \textbf{Value} \\
\hline \\
Number of layers $L$ & $6$ \\
Number of Fourier features $d_F$ & $6$ \\
Hidden dimensionality $d_H$ & $65$ \\
Latent variance $\sigma^2$ & $0.04$ \\
\end{tabular}
\end{table}

\subsubsection{Training Details}
\textit{Split-flows:} We train our model on the conditional flow matching objective presented in Section~\ref{sec:two-sided-fm}, using the coupling described in Section~\ref{sec:split-flows} and a linear reference interpolant:
\begin{equation}
I: [0, 1] \times \mathbb{R}^n \times \mathbb{R}^n \rightarrow \mathbb{R}^n, \qquad t, \vx_0, \vx_1 \mapsto I_t(\vx_0, \vx_1) = (1 - t)\vx_0 + t\vx_1.
\end{equation}
All training hyperparameters are listed in Table~\ref{tab:app-chignolin-benchmark-training}. Training takes approximately 40 hours on an NVIDIA A30 GPU with 30 GB of memory. 

\begin{table}[ht]
\caption{Training hyperparameters of the model trained for backmapping the $C_\alpha$ representation of chignolin. }  
\label{tab:app-chignolin-benchmark-training}
\centering
\begin{tabular}{lc}
\textbf{Hyperparameter} & \textbf{Value} \\
\hline \\
Optimizer & Adam \\
Learning rate (LR) & $3 \times 10^{-4}$ \\
LR scheduler & One-cycle\\
Weight decay & $1 \times 10^{-3}$\\
Batch size & 64 \\
Number of opt. steps & 24,760\\
\end{tabular}
\end{table}

\textit{TC-VAE:} We retrain the TC-VAE \cite{Shmilovich2022} using the code and hyperparameters provided by the authors. We find that the additional energy regularization proposed by the authors consistently leads to numerical instabilities, resulting in NaN values in the loss. We therefore train the model without the energy regularization. Training takes approximately 120 hours on an NVIDIA A30 GPU with 30 GB of memory.

\textit{CG-back:} Since CG-back \cite{UgarteLaTorre2025} is a transferable model, we utilize the pretrained model (M), provided by the authors.

\textit{Flow-back:} Flow-back \cite{Berlaga2025} is a transferable model. We hence use the pretrained model (Pro-pretrained) provided by the authors.

\subsubsection{Evaluation Metrics}
We evaluate and compare the backmapping capabilities of split-flows and reference methods using several metrics. In this section, we provide additional details on their computation. We will denote the reference configurations as $\vr$ and $\mR = M(\vr)$ and the reconstructed configurations as $\hat{\vr}$ and $\hat{\mR} = M(\hat{\vr})$ for the fine- and coarse-grained resolutions, respectively.

\textit{Wasserstein-1 distance of the internal energy distribution:} 
The internal potential energy of each configuration is computed using the AMBER14 all-atom force field under vacuum conditions. For each configuration $\vr$, we add hydrogen atoms using OpenMM and relax the positions via energy minimization (up to 1000 iterations or until the forces fall below $0.1~\unit{\kilo\joule\per\mol\per\nano\meter}$). We then compute the Wasserstein-1 distance between the distributions of internal energies of the reference and reconstructed configurations. In Figure~\ref{fig:app-energy-distribution}, we show histograms of the energy distributions.
\begin{figure}[ht]
    \centering
    \includegraphics[width=0.5\linewidth]{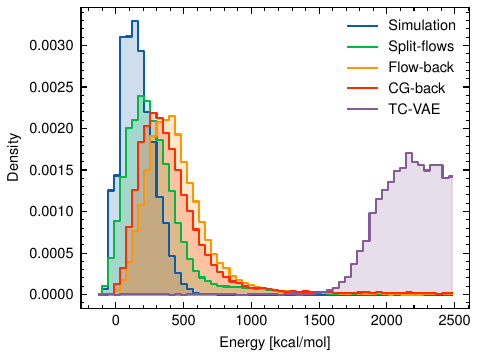}
    \caption{Distributions of internal energies of configurations from the reference simulation and backmapped configurations obtained from the methods in the comparison. The internal energy is evaluated using the AMBER14 all-atom force field.}
    \label{fig:app-energy-distribution}
\end{figure}

\textit{Coarse-grained RMSD:} 
To quantify the consistency between coarse-grained configurations and their corresponding coarse-grained representatives, we compute the root-mean-squared deviation (RMSD) between the $C_\alpha$ atoms of the coarse-grained and backmapped configurations. The RMSD is well-defined up to global rotations and translations, and can be expressed as:
\begin{equation} \label{eq:app-rmsd}
\mathrm{RMSD}_\mathrm{cg}(\mR, \hat{\mR}) = \min_{Q,\, \vt} \left[
\frac{1}{N} \sum_{i=1}^{N} 
\left\|Q\, \mR_i + \mathbf{t} - \hat{\mR}_i \right\|^2
\right]^{1/2},
\end{equation}
where $\mR_i$ and $\hat{\mR}_i$ denote the positions of the $i$-th coarse-grained particle for the reference and reconstructed configurations, respectively. Furthermore, $Q \in SO(3)$ is a global rotation matrix and $\vt \in \mathbb{R}^3$ a global translation vector.

\textit{Relative graph-edit distance:} 
We measure topological reconstruction quality using the relative graph edit distance between the molecular graph obtained from inter-atomic distances and the reference graph. Given a reconstructed configuration $\hat{\vr}$, we construct an adjacency matrix $\hat{A}$ based on the Van-der-Waals cutoff values $c$ in Table~\ref{tab:app-vdw-cutoff}, where the entry at position $ij$ is defined as:
\begin{equation}
    \hat{A}_{ij} = \begin{cases}
    1 & \text{if } \|\hat{\vr}_i - \hat{\vr}_j\|_2 < s (c_i + c_j), \\
    0 & \text{otherwise}
    \end{cases},
\end{equation}
where $c_i$ and $c_j$ are the respective cutoff values, and $s = 1.3$ is a scaling factor. We then compute the relative graph edit distance as:
\begin{equation}
D_\mathcal{G} = \frac{\sum_{ij} ( A - \hat{A} )_{ij}}{\sum_{ij} A_{ij}},
\end{equation}
where $A$ is the adjacency matrix of the reference molecular graph.
\begin{table}[ht]
\centering
\caption{VDW cutoff values in $\unit{\nano\meter}$ for atoms with atomic numbers 1--107.} \label{tab:app-vdw-cutoff}
\resizebox{1\linewidth}{!}{
\begin{tabular}{lclclclclclclclc}
\textbf{Z} & \textbf{Cutoff} & \textbf{Z} & \textbf{Cutoff} & \textbf{Z} & \textbf{Cutoff} & \textbf{Z} & \textbf{Cutoff} & \textbf{Z} & \textbf{Cutoff} & \textbf{Z} & \textbf{Cutoff} & \textbf{Z} & \textbf{Cutoff} & \textbf{Z} & \textbf{Cutoff} \\
\hline \\
1 & 0.023 & 2 & 0.093 & 3 & 0.068 & 4 & 0.035 & 5 & 0.083 & 6 & 0.068 & 7 & 0.068 & 8 & 0.068 \\
9 & 0.064 & 10 & 0.112 & 11 & 0.097 & 12 & 0.110 & 13 & 0.135 & 14 & 0.120 & 15 & 0.075 & 16 & 0.102 \\
17 & 0.099 & 18 & 0.157 & 19 & 0.133 & 20 & 0.099 & 21 & 0.144 & 22 & 0.147 & 23 & 0.133 & 24 & 0.135 \\
25 & 0.135 & 26 & 0.134 & 27 & 0.133 & 28 & 0.150 & 29 & 0.152 & 30 & 0.145 & 31 & 0.122 & 32 & 0.117 \\
33 & 0.121 & 34 & 0.122 & 35 & 0.121 & 36 & 0.191 & 37 & 0.147 & 38 & 0.112 & 39 & 0.178 & 40 & 0.156 \\
41 & 0.148 & 42 & 0.147 & 43 & 0.135 & 44 & 0.140 & 45 & 0.145 & 46 & 0.150 & 47 & 0.159 & 48 & 0.169 \\
49 & 0.163 & 50 & 0.146 & 51 & 0.146 & 52 & 0.147 & 53 & 0.140 & 54 & 0.198 & 55 & 0.167 & 56 & 0.134 \\
57 & 0.187 & 58 & 0.183 & 59 & 0.182 & 60 & 0.181 & 61 & 0.180 & 62 & 0.180 & 63 & 0.199 & 64 & 0.179 \\
65 & 0.176 & 66 & 0.175 & 67 & 0.174 & 68 & 0.173 & 69 & 0.172 & 70 & 0.194 & 71 & 0.172 & 72 & 0.157 \\
73 & 0.143 & 74 & 0.137 & 75 & 0.135 & 76 & 0.137 & 77 & 0.132 & 78 & 0.150 & 79 & 0.150 & 80 & 0.170 \\
81 & 0.155 & 82 & 0.154 & 83 & 0.154 & 84 & 0.168 & 85 & 0.170 & 86 & 0.240 & 87 & 0.200 & 88 & 0.190 \\
89 & 0.188 & 90 & 0.179 & 91 & 0.161 & 92 & 0.158 & 93 & 0.155 & 94 & 0.153 & 95 & 0.151 & 96 & 0.150 \\
97 & 0.150 & 98 & 0.150 & 99 & 0.150 & 100 & 0.150 & 101 & 0.150 & 102 & 0.150 & 103 & 0.150 & 104 & 0.157 \\
105 & 0.149 & 106 & 0.143 & 107 & 0.141 &  &  &  &  &  &  &  &  &  &  \\
\end{tabular}
}
\end{table}

\textit{Fiber diversity:} To measure the diversity of generated structures for a given coarse-grained representative, we draw $M$ samples on the fiber $\mathcal{F}(\mR)$ for a given coarse-grained representative $\mR = M(\vr)$. Following \cite{Jones2023}, we then define a diversity score $\eta_\mathrm{div}$ as the ratio of the average pairwise RMSD (see Equation~\ref{eq:app-rmsd}) between all generated configurations and the average RMSD between each generated structure and the reference configuration $\vr$:
\begin{equation}
    \eta_\mathrm{div} = \frac{\tfrac{2}{M(M-1)}\sum_{m \neq k} \mathrm{RMSD}(\vr_m, \vr_k)}{\tfrac{1}{M} \sum_{m} \mathrm{RMSD}(\vr_m, \vr)},
\end{equation}
where here $\vr_m$ and $\vr_k$ denote the $m$-th and $k$-th sample on the fiber.

\subsection{Solute in a Lipid Bilayer} \label{app:bilayer}
\subsubsection{Data Generation}
\textit{Simulation:} 
The simulated data is due to \cite{Nagel2025}, who simulate a coarse-grained POPC lipid bilayer interacting with a two-bead C1P3 solute using the Martini 3 force field. Simulations are performed in GROMACS 2024.3 with a time step of $0.02 \unit{\pico\second}$ and a simulation box of $6 \times 6 \times 10 \unit{\nano\meter^3}$ under periodic boundary conditions. The systems are first equilibrated for $200 \unit{\pico\second}$ in an NPT ensemble at $298 \unit{\kelvin}$ and $1 \unit{\bar}$. Subsequently, the system is simulated for $1 \unit{\micro\second}$ in an NVT ensemble with a constant biasing force of $10 \unit{\kilo\cal\per\mol\per\nano\meter}$ dragging the solute through the simulation box. Frames are saved every $0.2 \unit{\pico\second}$.

\textit{Data preparation:}
The simulated configurations are translated such that the membrane center is at the origin. We then extract the positions of the C1 and P3 beads from the simulated trajectory to compute a two-dimensional description consisting of the distance $z$ between the center of mass of the solute and the membrane center, and the relative orientation $\vartheta$ with respect to the $z$-axis.

\subsubsection{Model Parameterization}
\textit{Network architecture:} We parameterize the velocity field $\vv_t^\theta$ of the split-flow using a simple MLP. To account for the periodicity in the distance from the membrane center $z \in \left[-\tfrac{L}{2}, \tfrac{L}{2}\right]$, we apply a sine-cosine input parameterization to the MLP:
\begin{equation}
z \mapsto \begin{bmatrix} \sin\left( \tfrac{2\pi}{L} z \right) & \cos\left( \tfrac{2\pi}{L} z \right) \end{bmatrix}^T.
\end{equation}

\textit{Noise distribution:} As the target noise distribution $\pi_{\epsilon|R}$, we use a uniform distribution $\mathcal{U}([0, \pi])$ over the angular domain. With this choice, the flow directly provides access to the excess quantities, i.e., the excess local mapping entropy and the excess information loss.

We give our architectural hyperparameter choices in Table~\ref{tab:app-bilayer-architecture}.

\begin{table}[ht]
\caption{Architectural hyperparameters of the model trained for backmapping the reduced representation of a solute dragged through a lipid bilayer. We report the choices for the MLP parameterization of the velocity field.}  
\label{tab:app-bilayer-architecture}
\centering
\begin{tabular}{lc}
\textbf{Hyperparameter} & \textbf{Value} \\
\hline \\
Number of layers $L$ & $3$ \\
Hidden dimensionality $d_H$ & 32 \\
Activation function & ReLU \\
\end{tabular}
\end{table}

\subsubsection{Training Details}
Training takes approximately 40 minutes on an NVIDIA GeForce RTX 4060 with 8 GB of memory. We report all hyperparameter choices for training the model in Table~\ref{tab:app-bilayer-training}

\begin{table}[ht]
\caption{Training hyperparameters of the model trained for backmapping the reduced representation of a solute dragged through a lipid bilayer.}  
\label{tab:app-bilayer-training}
\centering
\begin{tabular}{lc}
\textbf{Hyperparameter} & \textbf{Value} \\
\hline \\
Optimizer & Adam \\
Learning rate (LR) & $1 \times 10^{-3}$ \\
LR scheduler & – \\
Weight decay & $0$\\
Batch size & 2048 \\
Number of opt. steps & 195,300\\
\end{tabular}
\end{table}

\subsubsection{KDE Comparison Details}
As a baseline, we fit a kernel density estimator (KDE) to the training data samples $\vr$ and $\mR$ to obtain a estimates $\pi_R^\mathrm{KDE}$ and $\pi_r^\mathrm{KDE}$ for the fine- and coarse-grained densities, respectively. We then bin the test data samples according to $\mR = \begin{bmatrix} z
\end{bmatrix}^\top$ to obtain an estimate of the local mapping entropy in Equation~\ref{eq:conditional-mapping-entropy}. In Figure~\ref{fig:lipid-bilayer} (B), we show the resulting excess information loss landscape across the simulation box. Furthermore, in Figure~\ref{fig:lipid-membrane-correlations}, we present a correlation plot between the split-flow and KDE estimates of the mapping entropy, which shows great agreement of the two methods, with a mean absolute error of $0.027$ and a Pearson correlation coefficient of $0.99$.
\begin{figure}[ht]
    \centering
    \includegraphics[width=0.5\linewidth]{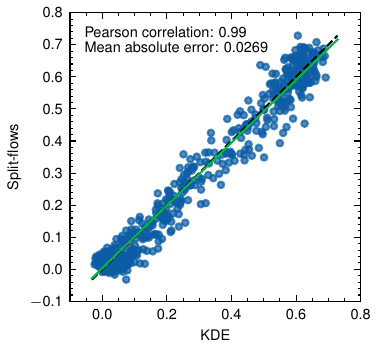}
    \caption{Correlation plot of split-flow and KDE estimates of the local information loss across the lipid membrane. The black dashed line denotes perfect agreement and nearly coincides with a linear fit to the data, shown in green.}
    \label{fig:lipid-membrane-correlations}
\end{figure}

\subsection{Alanine Dipeptide} \label{app:ala2}
\subsubsection{Data Generation}
\textit{Simulation:}
To obtain training data, we simulate alanine dipeptide using Langevin molecular dynamics (see Appendix~\ref{app:chignolin-data-generation}) in OpenMM with the AMBER14 all-atom force field and the TIP3P water model. The alanine dipeptide molecule is solvated in a cubic water box with $1 \unit{\nano\meter}$ padding and neutralized. Simulations are performed at $600 \unit{\kelvin}$ using Langevin dynamics ($1 \unit{\per\pico\second}$ friction, $2 \unit{\femto\second}$ timestep) with PME electrostatics and a $1 \unit{\nano\meter}$ cutoff. After energy minimization and velocity initialization, we simulate the system for $1 \unit{\micro\second}$ and save coordinates every $2 \unit{\pico\second}$.

\textit{Data preparation:}
To preprocess the raw simulation data, we remove the water solvent, retaining only the atoms of alanine dipeptide. We then center the coordinates of each individual configuration and superpose the configurations along the simulated trajectory.

\subsubsection{Rigidity of the Coarse-Grained Representation}
The visualization of the local mapping entropy in the $(\phi, \psi)$ plane of Ramachandran angles relies on the rigidity of the bonds and angles between the five atoms in the coarse-grained representation. In Table~\ref{tab:app-bond-angle-deviations}, we report the mean relative deviations of bond lengths and bond angles. We observe only small bond and angular fluctuations up to approximately $3\%$, indicating that these degrees of freedom contribute negligibly to the configurational entropy.

\begin{table}[ht]
\caption{Relative deviation of the bond lengths and bond angles to the their respective mean values. We report mean and standard deviation evaluated over trajectory used for training.}
\label{tab:app-bond-angle-deviations}
\centering
\begin{tabular}{lc}
\textbf{Bond / Angle} & \textbf{Relative deviation [$\%$]} \\
\hline \\
Bond C--N & \valpm{2.07}{1.56} \\
Bond N--CA & \valpm{2.20}{1.64} \\
Bond CA--C & \valpm{2.14}{1.64} \\
Bond C--N & \valpm{2.01}{1.57} \\
Angle C--N--CA & \valpm{2.90}{2.20} \\
Angle N--CA--C & \valpm{3.20}{2.43} \\
Angle C--C--N & \valpm{2.73}{2.07} \\
\end{tabular}
\end{table}

\subsubsection{Model Parameterization}
We parameterize the model analogously to the model trained on chignolin, as presented in Appendix~\ref{app:chignolin-model-parameterization}. We give our hyperparameter choices in Table~\ref{tab:app-ala2-benchmark-architecture}.

\begin{table}[ht]
\caption{Architectural hyperparameters of the model trained for backmapping the coarse-grained representation of alanine dipeptide. We report the choices for the $E(3)$-equivariant GNN parameterization of the velocity field and the noise distribution.}  
\label{tab:app-ala2-benchmark-architecture}
\centering
\begin{tabular}{lc}
\textbf{Hyperparameter} & \textbf{Value} \\
\hline \\
Number of layers $L$ & $4$ \\
Number of Fourier features $d_F$ & $6$ \\
Hidden dimensionality $d_H$ & $129$ \\
Latent variance $\sigma^2$ & $0.04$ \\
\end{tabular}
\end{table}

\subsubsection{Training Details}
Training takes about 18 hours on an NVIDIA A30 GPU with 30 GB of memory. We report all hyperparameter choices for training in Table~\ref{tab:app-ala2-training}

\begin{table}[ht]
\caption{Training hyperparameters of the model trained for backmapping the coarse-grained representation of alanine dipeptide.}  
\label{tab:app-ala2-training}
\centering
\begin{tabular}{lc}
\textbf{Hyperparameter} & \textbf{Value} \\
\hline \\
Optimizer & Adam \\
Learning rate (LR) & $1 \times 10^{-4}$ \\
LR scheduler & Exponential ($\gamma=0.999$) \\
Weight decay & $0$\\
Batch size & 64 \\
Number of opt. steps & 50,500\\
\end{tabular}
\end{table}

\vfill

%% file: literature.bib
@article{Prez-Hernndez2013,
   author = {Guillermo Pérez-Hernández and Fabian Paul and Toni Giorgino and Gianni De Fabritiis and Frank Noé},
   doi = {10.1063/1.4811489/192538},
   issn = {00219606},
   issue = {1},
   journal = {Journal of Chemical Physics},
   month = {7},
   pmid = {23822324},
   publisher = {AIP Publishing},
   title = {Identification of slow molecular order parameters for Markov model construction},
   volume = {139},
   url = {/aip/jcp/article/139/1/015102/192538/Identification-of-slow-molecular-order-parameters},
   year = {2013}
}

@article{Noid2013,
   author = {W. G. Noid},
   doi = {10.1063/1.4818908},
   issn = {0021-9606},
   issue = {9},
   journal = {The Journal of Chemical Physics},
   month = {9},
   title = {Perspective: Coarse-grained models for biomolecular systems},
   volume = {139},
   year = {2013}
}

@inproceedings{Brehmer2020,
    author = {Brehmer, Johann and Cranmer, Kyle},
    booktitle = {Advances in Neural Information Processing Systems},
    editor = {H. Larochelle and M. Ranzato and R. Hadsell and M.F. Balcan and H. Lin},
    pages = {442--453},
    publisher = {Curran Associates, Inc.},
    title = {Flows for simultaneous manifold learning and density estimation},
    url = {https://proceedings.neurips.cc/paper_files/paper/2020/file/051928341be67dcba03f0e04104d9047-Paper.pdf},
    volume = {33},
    year = {2020}
}

@Inbook{Lee2003,
author="Lee, John M.",
title="Smooth Manifolds",
bookTitle="Introduction to Smooth Manifolds",
year="2003",
publisher="Springer New York",
address="New York, NY",
pages="1--29",
isbn="978-0-387-21752-9",
doi="10.1007/978-0-387-21752-9_1",
url="https://doi.org/10.1007/978-0-387-21752-9_1"
}

@article{Noid2023,
   author = {W. G. Noid},
   doi = {10.1021/ACS.JPCB.2C08731/ASSET/IMAGES/LARGE/JP2C08731_0006.JPEG},
   issn = {15205207},
   issue = {19},
   journal = {Journal of Physical Chemistry B},
   month = {5},
   pages = {4174-4207},
   pmid = {37149781},
   publisher = {American Chemical Society},
   title = {Perspective: Advances, Challenges, and Insight for Predictive Coarse-Grained Models},
   volume = {127},
   url = {https://pubs.acs.org/doi/full/10.1021/acs.jpcb.2c08731},
   year = {2023}
}

@article{Nagel2025,
   author = {Daniel Nagel and Tristan Bereau},
   month = {6},
   title = {Fokker-Planck Score Learning: Efficient Free-Energy Estimation under Periodic Boundary Conditions},
   year = {2025},
   journal = {arXiv preprint arXiv:2506.15653}
}

@article{UgarteLaTorre2025,
   author = {Diego Ugarte La Torre and Yuji Sugita},
   doi = {10.1021/acs.jcim.5c01281},
   issn = {1549-9596},
   journal = {Journal of Chemical Information and Modeling},
   month = {9},
   title = {CGBack: Diffusion Model for Backmapping Large-Scale and Complex Coarse-Grained Molecular Systems},
   year = {2025}
}

@InProceedings{Satorras2021,
  title = 	 {E(n) Equivariant Graph Neural Networks},
  author =       {Satorras, V\'{\i}ctor Garcia and Hoogeboom, Emiel and Welling, Max},
  booktitle = 	 {Proceedings of the 38th International Conference on Machine Learning},
  pages = 	 {9323--9332},
  year = 	 {2021},
  editor = 	 {Meila, Marina and Zhang, Tong},
  volume = 	 {139},
  series = 	 {Proceedings of Machine Learning Research},
  month = 	 {18--24 Jul},
  publisher =    {PMLR},
  url = 	 {https://proceedings.mlr.press/v139/satorras21a.html}
}

@article{Wassenaar2014,
   author = {Tsjerk A. Wassenaar and Kristyna Pluhackova and Rainer A. Böckmann and Siewert J. Marrink and D. Peter Tieleman},
   doi = {10.1021/ct400617g},
   issn = {15499618},
   issue = {2},
   journal = {Journal of Chemical Theory and Computation},
   title = {Going backward: A flexible geometric approach to reverse transformation from coarse grained to atomistic models},
   volume = {10},
   year = {2014}
}

@article{Rzepiela2010,
   author = {Andrzej J. Rzepiela and Lars V. Schäfer and Nicolae Goga and H. Jelger Risselada and Alex H. De Vries and Siewert J. Marrink},
   doi = {10.1002/jcc.21415},
   issn = {01928651},
   issue = {6},
   journal = {Journal of Computational Chemistry},
   title = {Software news and update reconstruction of atomistic details from coarse-grained structures},
   volume = {31},
   year = {2010}
}

@article{Peter2009,
   author = {Christine Peter and Kurt Kremer},
   doi = {10.1039/b912027k},
   issn = {1744683X},
   issue = {22},
   journal = {Soft Matter},
   title = {Multiscale simulation of soft matter systems - From the atomistic to the coarse-grained level and back},
   volume = {5},
   year = {2009}
}

@InProceedings{Tong2024,
    title = 	 {Simulation-Free {S}chrödinger Bridges via Score and Flow Matching},
    author =       {Tong, Alexander Y. and Malkin, Nikolay and Fatras, Kilian and Atanackovic, Lazar and Zhang, Yanlei and Huguet, Guillaume and Wolf, Guy and Bengio, Yoshua},
    booktitle = 	 {Proceedings of The 27th International Conference on Artificial Intelligence and Statistics},
    pages = 	 {1279--1287},
    year = 	 {2024},
    editor = 	 {Dasgupta, Sanjoy and Mandt, Stephan and Li, Yingzhen},
    volume = 	 {238},
    series = 	 {Proceedings of Machine Learning Research},
    month = 	 {02--04 May},
    publisher =    {PMLR},
    url = 	 {https://proceedings.mlr.press/v238/tong24a.html}
}

@article{Jones2023,
   author = {Michael S. Jones and Kirill Shmilovich and Andrew L. Ferguson},
   doi = {10.1021/acs.jctc.3c00840},
   issn = {15499626},
   issue = {21},
   journal = {Journal of Chemical Theory and Computation},
   title = {DiAMoNDBack: Diffusion-Denoising Autoregressive Model for Non-Deterministic Backmapping of Cα Protein Traces},
   volume = {19},
   year = {2023}
}

@article{Armstrong2012,
   author = {J. A. Armstrong and C. Chakravarty and P. Ballone},
   doi = {10.1063/1.3697383},
   issn = {00219606},
   issue = {12},
   journal = {Journal of Chemical Physics},
   title = {Statistical mechanics of coarse graining: Estimating dynamical speedups from excess entropies},
   volume = {136},
   year = {2012}
}

@article{Jin2023,
   author = {Jaehyeok Jin and Kenneth S. Schweizer and Gregory A. Voth},
   doi = {10.1063/5.0116299},
   issn = {10897690},
   issue = {3},
   journal = {Journal of Chemical Physics},
   title = {Understanding dynamics in coarse-grained models. I. Universal excess entropy scaling relationship},
   volume = {158},
   year = {2023}
}

@article{Giulini2020,
   author = {Marco Giulini and Roberto Menichetti and M. Scott Shell and Raffaello Potestio},
   doi = {10.1021/acs.jctc.0c00676},
   issn = {15499626},
   issue = {11},
   journal = {Journal of Chemical Theory and Computation},
   title = {An Information-Theory-Based Approach for Optimal Model Reduction of Biomolecules},
   volume = {16},
   year = {2020}
}

@article{Mussi2025,
   author = {Lucus M. Mussi and W. G. Noid},
   doi = {10.1063/5.0281108},
   issn = {0021-9606},
   issue = {8},
   journal = {The Journal of Chemical Physics},
   month = {8},
   title = {Predicting energetic and entropic driving forces with coarse-grained models},
   volume = {163},
   year = {2025}
}

@article{Holtzman2022,
   author = {Roi Holtzman and Marco Giulini and Raffaello Potestio},
   doi = {10.1103/PhysRevE.106.044101},
   issn = {2470-0045},
   issue = {4},
   journal = {Physical Review E},
   month = {10},
   pages = {044101},
   title = {Making sense of complex systems through resolution, relevance, and mapping entropy},
   volume = {106},
   year = {2022}
}

@inproceedings{rezende2015flows,
   author = {Danilo Jimenez Rezende and Shakir Mohamed},
   booktitle = {32nd International Conference on Machine Learning, ICML 2015},
   title = {Variational inference with normalizing flows},
   volume = {2},
   year = {2015}
}

@article{Shmilovich2022,
   author = {Kirill Shmilovich and Marc Stieffenhofer and Nicholas E. Charron and Moritz Hoffmann},
   doi = {10.1021/acs.jpca.2c07716},
   issn = {1089-5639},
   issue = {48},
   journal = {The Journal of Physical Chemistry A},
   month = {12},
   pages = {9124-9139},
   title = {Temporally Coherent Backmapping of Molecular Trajectories From Coarse-Grained to Atomistic Resolution},
   volume = {126},
   year = {2022}
}

@InProceedings{Wang2022,
    title = 	 {Generative Coarse-Graining of Molecular Conformations},
    author =       {Wang, Wujie and Xu, Minkai and Cai, Chen and Miller, Benjamin K and Smidt, Tess and Wang, Yusu and Tang, Jian and Gomez-Bombarelli, Rafael},
    booktitle = 	 {Proceedings of the 39th International Conference on Machine Learning},
    pages = 	 {23213--23236},
    year = 	 {2022},
    editor = 	 {Chaudhuri, Kamalika and Jegelka, Stefanie and Song, Le and Szepesvari, Csaba and Niu, Gang and Sabato, Sivan},
    volume = 	 {162},
    series = 	 {Proceedings of Machine Learning Research},
    month = 	 {17--23 Jul},
    publisher =    {PMLR},
    url = 	 {https://proceedings.mlr.press/v162/wang22ag.html}
}

@article{Li2020,
   author = {Wei Li and Craig Burkhart and Patrycja Polińska and Vagelis Harmandaris and Manolis Doxastakis},
   doi = {10.1063/5.0012320},
   issn = {0021-9606},
   issue = {4},
   journal = {The Journal of Chemical Physics},
   month = {7},
   title = {Backmapping coarse-grained macromolecules: An efficient and versatile machine learning approach},
   volume = {153},
   year = {2020}
}

@article{Stieffenhofer2020,
   author = {Marc Stieffenhofer and Michael Wand and Tristan Bereau},
   doi = {10.1088/2632-2153/abb6d4},
   issn = {2632-2153},
   issue = {4},
   journal = {Machine Learning: Science and Technology},
   month = {12},
   pages = {045014},
   title = {Adversarial reverse mapping of equilibrated condensed-phase molecular structures},
   volume = {1},
   year = {2020}
}

@article{Berlaga2025,
   author = {Alex Berlaga and Michael S. Jones and Andrew L. Ferguson},
   month = {8},
   title = {FlowBack-Adjoint: Physics-Aware and Energy-Guided Conditional Flow-Matching for All-Atom Protein Backmapping},
   year = {2025},
   journal = {arXiv preprint arXiv:2508.03619}
}

@article{ferguson2025flowback,
   author = {Michael S. Jones and Smayan Khanna and Andrew L. Ferguson},
   doi = {10.1021/acs.jcim.4c02046},
   issn = {1549-9596},
   issue = {2},
   journal = {Journal of Chemical Information and Modeling},
   month = {1},
   pages = {672-692},
   title = {FlowBack: A Generalized Flow-Matching Approach for Biomolecular Backmapping},
   volume = {65},
   year = {2025}
}

@InProceedings{Albergo2023data,
    title = 	 {Stochastic Interpolants with Data-Dependent Couplings},
    author =       {Albergo, Michael Samuel and Goldstein, Mark and Boffi, Nicholas Matthew and Ranganath, Rajesh and Vanden-Eijnden, Eric},
    booktitle = 	 {Proceedings of the 41st International Conference on Machine Learning},
    pages = 	 {921--937},
    year = 	 {2024},
    editor = 	 {Salakhutdinov, Ruslan and Kolter, Zico and Heller, Katherine and Weller, Adrian and Oliver, Nuria and Scarlett, Jonathan and Berkenkamp, Felix},
    volume = 	 {235},
    series = 	 {Proceedings of Machine Learning Research},
    month = 	 {21--27 Jul},
    publisher =    {PMLR},
    url = 	 {https://proceedings.mlr.press/v235/albergo24a.html}
}

@article{Giulini2024,
   author = {Marco Giulini and Raffaele Fiorentini and Luca Tubiana and Raffaello Potestio and Roberto Menichetti},
   doi = {10.1021/ACS.JCIM.4C00490/ASSET/IMAGES/LARGE/CI4C00490_0007.JPEG},
   issn = {1549960X},
   issue = {12},
   journal = {Journal of Chemical Information and Modeling},
   month = {6},
   pages = {4912-4927},
   pmid = {38860513},
   publisher = {American Chemical Society},
   title = {EXCOGITO, an Extensible Coarse-Graining Toolbox for the Investigation of Biomolecules by Means of Low-Resolution Representations},
   volume = {64},
   url = {https://pubs.acs.org/doi/full/10.1021/acs.jcim.4c00490},
   year = {2024}
}

@article{Kidder2024,
   author = {Katherine M. Kidder and W. G. Noid},
   doi = {10.1063/5.0220989},
   issn = {0021-9606},
   issue = {13},
   journal = {The Journal of Chemical Physics},
   month = {10},
   title = {Analysis of mapping atomic models to coarse-grained resolution},
   volume = {161},
   year = {2024}
}

@article{Kidder2021,
   author = {Katherine M. Kidder and Ryan J. Szukalo and W. G. Noid},
   doi = {10.1140/epjb/s10051-021-00153-4},
   issn = {1434-6028},
   issue = {7},
   journal = {The European Physical Journal B},
   month = {7},
   pages = {153},
   title = {Energetic and entropic considerations for coarse-graining},
   volume = {94},
   year = {2021}
}

@article{Pimentel2025a,
   author = {João V.M. Pimentel and Vladimir A. Mandelshtam},
   doi = {10.1063/5.0273835},
   issn = {10897690},
   issue = {21},
   journal = {Journal of Chemical Physics},
   title = {From all-atom to rigid monomer treatment of molecular clusters},
   volume = {162},
   year = {2025}
}

@article{Foley2015,
   author = {Thomas T. Foley and M. Scott Shell and W. G. Noid},
   doi = {10.1063/1.4929836},
   issn = {0021-9606},
   issue = {24},
   journal = {The Journal of Chemical Physics},
   month = {12},
   title = {The impact of resolution upon entropy and information in coarse-grained models},
   volume = {143},
   year = {2015}
}

@article{shell2008entropy,
   author = {M. Scott Shell},
   doi = {10.1063/1.2992060},
   issn = {0021-9606},
   issue = {14},
   journal = {The Journal of Chemical Physics},
   month = {10},
   title = {The relative entropy is fundamental to multiscale and inverse thermodynamic problems},
   volume = {129},
   year = {2008}
}

@article{Rudzinski2011,
   author = {Joseph F. Rudzinski and W. G. Noid},
   doi = {10.1063/1.3663709},
   issn = {00219606},
   issue = {21},
   journal = {Journal of Chemical Physics},
   title = {Coarse-graining entropy, forces, and structures},
   volume = {135},
   year = {2011}
}

@article{Albergo2023,
   author = {Michael S. Albergo and Nicholas M. Boffi and Eric Vanden-Eijnden},
   month = {3},
   title = {Stochastic Interpolants: A Unifying Framework for Flows and Diffusions},
   url = {https://arxiv.org/abs/2303.08797v3},
   year = {2023},
   journal = {arXiv preprint arXiv:2303.08797},
}

@inproceedings{Chen2018,
    author = {Chen, Ricky T. Q. and Rubanova, Yulia and Bettencourt, Jesse and Duvenaud, David K},
    booktitle = {Advances in Neural Information Processing Systems},
    editor = {S. Bengio and H. Wallach and H. Larochelle and K. Grauman and N. Cesa-Bianchi and R. Garnett},
    pages = {},
    publisher = {Curran Associates, Inc.},
    title = {Neural Ordinary Differential Equations},
    url = {https://proceedings.neurips.cc/paper_files/paper/2018/file/69386f6bb1dfed68692a24c8686939b9-Paper.pdf},
    volume = {31},
    year = {2018}
}

@inproceedings{
    Lipman2023,
    title={Flow Matching for Generative Modeling},
    author={Yaron Lipman and Ricky T. Q. Chen and Heli Ben-Hamu and Maximilian Nickel and Matthew Le},
    booktitle={The Eleventh International Conference on Learning Representations },
    year={2023},
    url={https://openreview.net/forum?id=PqvMRDCJT9t}
}
